\documentclass[a4paper,11pt]{article}

\usepackage{natbib} 

\usepackage{amssymb,type1cm}

\usepackage[pdftex,bookmarks=true,bookmarksopen=true]{hyperref}

\usepackage[pdftex]{graphicx}

\title{Preference aggregation theory without acyclicity: 
The core without majority dissatisfaction\thanks{
Preprint, Games and Economic Behavior (2011) 72:187-201,\protect\\
\href{http://dx.doi.org/10.1016/j.geb.2010.06.008}{doi:10.1016/j.geb.2010.06.008}.
}}

\author{Masahiro Kumabe \\
Faculty of Liberal Arts, The Open University of Japan\\
2-11 Wakaba, Mihama-ku, Chiba City, 261-8586 Japan
\and 
H.  Reiju Mihara\thanks{Corresponding author. \protect\\
\emph{URL:} \url{http://econpapers.repec.org/RAS/pmi193.htm} (H.R. Mihara).}\\
Kagawa University Library\\
Takamatsu 760-8525, Japan}

\date{July 2010}

\newcommand{\B}{\mathcal{B}}
\newcommand{\W}{\mathcal{W}}
\newcommand{\Y}{\mathcal{Y}}
\newcommand{\Z}{\mathcal{Z}}

\newcommand{\A}{\mathcal{A}}
\newcommand{\M}{\mathcal{M}}
\newcommand{\p}{\mathbf{p}}
\newcommand{\pref}{\succ_i^\p}
\newcommand{\profile}{\mbox{$(\pref)_{i\in N}$}}
\newcommand{\dom}{\succ^\p_\W}

\newcommand{\xprefy}{\{\,i: {x\pref y}\,\}}
\newcommand{\yprefx}{\{\,i: {y\pref x}\,\}}
\newcommand{\xnotmax}{ \{\,i: x\notin\max_B\pref \,\}}
\newcommand{\profs}{\M(B)^{N}_\B}
\newcommand{\acyclicprofs}{\A^N_\B}

\newcommand{\toto}{\to\!\to}
\newcommand{\qed}{\enspace\enspace \vrule height 6pt width5pt
depth2pt}

\usepackage{theorem}

\newtheorem{theorem}{Theorem}
\newtheorem{prop}{Proposition}
\newtheorem{cor}{Corollary}
\newtheorem{lemma}{Lemma}

{\theorembodyfont{\normalfont}
\newtheorem{definition}{Definition}
\newtheorem{example}{Example}
\newtheorem{remark}{Remark}
}

\newenvironment{proof}{\emph{Proof}.}{\qed\bigskip}

\begin{document} 

\maketitle 


\begin{abstract} 
Acyclicity of individual preferences is a minimal assumption in social choice theory. 
We replace that assumption by the direct assumption that 
preferences have maximal elements on a fixed agenda.
We show that the core of a simple game is nonempty 
for all profiles of such preferences 
if and only if the number of alternatives in the agenda is less than 
the Nakamura number of the game.
The same is true if we replace the core by the \emph{core without majority dissatisfaction},
 obtained by deleting from the agenda 
all the alternatives that are non-maximal for all players in a winning coalition.
Unlike the core, the core without majority dissatisfaction depends only on
the players' sets of maximal elements and is included in the union of such sets.
A result for an extended framework gives another sense in which 
the core without majority dissatisfaction behaves better than the core.

\emph{Journal of Economic Literature} Classifications:  C71, D71, C02.

\end{abstract}

\pagebreak

\section{Introduction}

\subsection{Preference aggregation theory for acyclic individual preferences}

\emph{Preference aggregation theory} is concerned with aggregating individual preferences
into a (collective) social preference, which is then maximized to yield a set of best alternatives.
The theory investigates the extent to which 
social preferences inherit desirable properties from individual preferences.
We typically restrict  (strict) individual and social preferences to
 those asymmetric relations $\succ$ on a set~$X$
of alternatives that are either (i)~acyclic or (ii)~transitive or (iii)~negatively transitive.\footnote{
Define the weak preference~$\succeq$ by 
$x\succeq y \Leftrightarrow y\not\succ x$.
$\succ$ is asymmetric iff $\succeq$ is complete (reflexive and total).
(i)~$\succ$ is \emph{acyclic} if for any finite set $\{x_1, x_2, \ldots, x_m\}\subseteq X$,
whenever $x_1 \succ x_2$, \ldots, $x_{m-1} \succ x_m$, we have $x_m \not\succ x_1$.
If $\succ$ is acyclic, it is asymmetric and irreflexive.
(ii)~$\succ$ is \emph{transitive} if $x\succ y$ and $y\succ z$ imply $x\succ z$.
When $\succ$ is transitive, we say $\succeq$ is \emph{quasi-transitive}.
(iii)~$\succ$ is \emph{negatively transitive} if $x\not\succ y$ and $y\not\succ z$ imply
$x\not\succ z$.
$\succ$ is negatively transitive iff $\succeq$ is transitive.} %

Of the properties (i), (ii), and (iii) for asymmetric preferences,
\emph{negative transitivity} is the most demanding.
Arrow's Theorem (\citeyear{arrow63}) points out the difficulty of aggregating 
preferences for more than two alternatives 
while preserving asymmetry and negative transitivity.\footnote{
The restriction to two alternatives disappears when there are infinitely many individuals \citep{fishburn70},
but such a resolution relies on highly nonconstructive mathematical objects \citep{mihara97et,mihara99jme}.} %

\emph{Acyclicity} is the least demanding of the properties;
it is necessary and sufficient 
for the existence of a maximal element on \emph{every finite subset} of alternatives.
The \emph{Nakamura number} plays a critical role in the study of 
preference aggregation rules with \emph{acyclic} social preferences.\footnote{
 \citet{banks95}, \citet{truchon95}, \citet{andjiga-m00}, and 
 \citet{kumabe-m08jme} are recent contributions to the literature.
 Earlier papers on acyclic rules can be found in \citet{truchon95} and \citet{austensmith-b99}.
 \citet{kumabe-m08scw} comprehensively study the restrictions that various properties for a simple game 
 impose on its Nakamura number.} %
Consider a \emph{simple game} (voting game) $\W$, a collection of ``winning'' coalitions in a set~$N$ of players.
Combining the game with a set $X$ of alternatives and a profile $\p=\profile$ of individual preferences, 
one obtains a \emph{simple game with preferences} $(\W,X,\p)$, for which one can define 
the social preference~$\dom$ (dominance relation) and
 the \emph{core} $C(\W,X,\p)$ (the set of maximal alternatives with respect to $\dom$).
Nakamura's theorem (\citeyear{nakamura79}) gives a restriction
on the number of alternatives that the set of players can deal with rationally (Theorem~\ref{nakamura-thm}):
the core of a simple game with preferences is always (i.e., for all profiles of acyclic preferences) nonempty
 if and only if
the number of alternatives is finite and below a certain number,
 called the \emph{Nakamura number} of the simple game.
 The theorem thus gives a condition (Corollary~\ref{nakamura-cor}) for the social preferences
 to inherit acyclicity from individual preferences.
 
To deal with an empty core (or cycles in social preferences), 
several authors have investigated solutions different from the core.
We pick two examples for which there have been recent developments.
First, \citet{duggan07} proposes a procedure in which one deletes
particular instances of preferences until the resulting subrelation is acyclic (alternatively, transitive or negatively
transitive), and collects the maximal elements of all such (maximal acyclic) subrelations.\footnote{%
A related procedure is to collect the maximal elements of all maximal \emph{chains} (subsets of alternatives 
on which the majority preference is a linear order), 
which yields the Banks set \citep{banks85,penn06scw27p531}; the set
consists of the sophisticated voting outcomes of some binary agenda.} %
Second, taking voters' foresight into consideration, 
\citet{rubinstein80} proposes the \emph{stability set}, a superset of the core.\footnote{%
\citet{lebreton-s90} provide a sufficient condition for the general nonemptyness of the stability set
in terms of the Nakamura number.
Using more complex characteristic numbers, \citet{martin-m06} propose a weaker sufficient condition
and \citet{andjiga-m06} a necessary and sufficient condition for 
the general nonemptyness of the stability set.}
All these investigations focus on treating cyclic \emph{social} preferences,
assuming that \emph{individual} preferences satisfy a rationality property at least as strong as acyclicity.

\subsection{Preference aggregation theory without acyclicity}
 
In this paper, we propose a \emph{preference aggregation theory without acyclicity}.
In contrast to the authors cited in the preceding paragraph,
we do not attempt to remove cycles in (social and even individual) preferences.
 
We retain the usual framework \citep[e.g.,][Section II.2]{arrow63}
 in which a set $X$ of underlying alternatives is distinguished from an
  \emph{agenda} (opportunity set) $B\subseteq X$ with which a group~$N$ of players are confronted.
In particular, 
fixing a simple game~$\W$ and a set $X$, we focus on
the aggregation methods that assign the core $C(\W,B,\p)$ or
the \emph{core without majority dissatisfaction} $C^+(\W,B,\p)$ (introduced later)\footnote{
As the notation suggests, we define the \emph{core (without majority dissatisfaction)} relative to an agenda~$B$.}
to each $(\W,B,\p)$ satisfying a certain assumption.
The assumption, which is actually a condition concerning a pair $(B,\p)$, is that 
\emph{every player's preference $\pref$ has a maximal element in~$B$}: 
the \emph{maximal set} $\max_B\pref$ (the set of maximal elements of the preference)
 is nonempty for each~$i$.\footnote{
Acyclicity of a preference is independent of the property of having
a maximal element.  There is a cyclic preference that has a maximal element on~$X$.
When $X$ is infinite, there is an acyclic preference that has no maximal element on~$X$.
Similarly, the property of having a maximal element on~$B$ is independent of the property of
having a maximal element on~$X$.} %
This is a rather direct assumption at the individual level, 
which is to be inherited to the requirement at the social level that 
there is a chosen alternative for the given pair.

The assumption is distinctive in that it involves an agenda~$B$ as well as a profile~$\p$.
For this reason, it does not fit the following \emph{Standard Scenario}
in social choice theory \citep[page~12]{arrow63} very well:
before knowing an agenda, the players discover and report their own preferences (on~$X$) to the planner;
having obtained a choice rule that assigns a nonempty subset to every agenda,
the planner applies the rule to a particular agenda~$B$.
Since the planner does not generally know whether a pair $(B,\p)$ satisfies the assumption
 until she faces the agenda~$B$,
what she obtains immediately after learning the profile~$\p$ of preferences is 
only a \emph{partial} choice rule, which might assign an empty set to some agenda.\footnote{
	This is not to say that \emph{partial} rules are uninteresting as an object of study.
	Computability theory \citep[e.g.,][]{odifreddi92}, for example, is powerful precisely 
	because it includes partial functions in its scope.} %

An \emph{Alternative Scenario} that the assumption fits well is the following:
the planner first presents an agenda~$B$ to the players, who then discover and report their own preferences 
(for alternatives in~$B$ at least) to the planner; the planner then makes a choice.
This is perhaps a closer description of actual collective decisions.
While failing to produce a choice rule at an intermediate stage,
the scenario has some advantages over the Standard Scenario.

First, the Alternative Scenario can deal with context-dependent choices based on multiple \emph{rationales}
(preferences belonging to the same individual) more easily,
where the context is given by an agenda \citep{kalai-rs02,ambrus-r0807}.
The problem with the Standard Scenario is that a player is supposed to report a single preference
 for the whole set~$X$, when she might actually have an irreducible set of multiple rationales.
 
Second, as \citet[page~110]{arrow63} writes, ``ideally, one could observe all preferences among
the available alternatives, but there would be no way to observe preferences among alternatives
not feasible for society,'' even if each player has a single preference.
This argument justifies the Alternative Scenario either directly or via
\emph{Arrow's IIA} (Independence of Irrelevant Alternatives), 
which requires that social choice from an agenda
depends only on the individual preferences restricted to the agenda.
Both the core and the core without majority dissatisfaction satisfy Arrow's IIA.\footnote{
To be precise, whenever
 $\pref\cap (B\times B)=\, \succ_i^{\p'}\cap (B\times B)$ for all $i$,
 we have $C(\W,B,\p)=C(\W,B,\p')$ and $C^+(\W,B,\p)=C^+(\W,B,\p')$.
The assertion for $C^+$ is a corollary of Lemma~\ref{c+maximal}.
The reader should not confuse Arrow's IIA with the IIA (sometimes called \emph{property~$\alpha$}) 
for choice rules, discussed by \citet{kalai-rs02}, for example.} %
\footnote{
Because of Arrow's IIA, it does not really matter whether we define preferences on~$X$ or on~$B$.
If a preference were defined on an agenda, however, a more straightforward formulation 
would be to remove the symbol~$B$ from the framework and call that agenda~$X$.
Doing so, however, would make the connection between the Alternative Scenario and the framework much less clear.
For this reason, we \emph{formally} define preferences on~$X$ instead of on~$B$.}

\subsection{The core without majority dissatisfaction}

The \emph{core without majority dissatisfaction} (Definition~\ref{c+})
is obtained by deleting from an agenda all the alternatives that are non-maximal for 
all \emph{individual} players in a large (winning) coalition.
It is (Lemma~\ref{c+in_c}) a strengthening or subset of the \emph{core},
obtained by deleting from the agenda all the alternatives that are non-maximal for a large (winning) 
coalition of players \emph{collectively}.
Consequently, it only chooses Pareto efficient alternatives from an agenda,
unlike other solutions such as the stability set \citep[page~153]{rubinstein80}.

The core without majority dissatisfaction is a simple solution concept
 that treats the maximal sets $\max_B \pref$ 
of the players in a better-behaved way than the core does.
(It is reasonable to pay attention to such sets, since they are the very objects that 
we assume to be nonempty.)
First, unlike the core, the core without majority dissatisfaction
 depends only on the players' maximal sets (Lemma~\ref{c+maximal}).\footnote{%
This property is sometimes called ``tops-only''; it is investigated in an abstract social choice
framework by \citet{mihara00scw}, for example.}
Second, each alternative in the core without majority dissatisfaction belongs to
someone's maximal set (Lemma~\ref{c+in_maximals}).
The same is not true for the core (Examples~\ref{ex:core1} and \ref{ex:core2}).

\subsection{Overview of the results}

The main results of the paper are similar to Nakamura's theorem (\citeyear{nakamura79}),
except that they consider \emph{profiles for} an agenda~$B$---profiles of (not necessarily acyclic)
preferences that have maximal elements on the agenda.

\emph{The first main result}, Theorem~\ref{nakamura-max}, is about the original Nakamura number
for simple games $\W$ defined on an algebra of coalitions.
It asserts that the following statements are equivalent:
(i)~the number of alternatives in the agenda~$B$ is less than the Nakamura number~$\nu(\W)$;
(ii)~the core $C^+(\W,B,\p)$ without majority dissatisfaction is nonempty for all profiles $\p$ for the agenda;
(iii)~the core $C(\W,B,\p)$ is nonempty for all profiles $\p$ for the agenda.
Regardless of which choice rule is used, the Nakamura number therefore measures
the extent (the size of an agenda) to which simple games will assuredly choose some alternative 
from the agenda,
whether individual preferences are assumed to be acyclic or to have maximal elements.

Theorem~\ref{nakamura-max} is remarkable for the following reasons:
First, it demonstrates that one can obtain a significant result without assuming acyclic preferences.
Though neglected in the literature, a preference aggregation theory without acyclicity has potential.
Second, the \emph{general} nonemptyness of the core implies the \emph{general} nonemptyness of the 
\emph{strengthening} of the core.  
The core without majority dissatisfaction is as satisfactory as the core according to this criterion.
Third, restricting preferences to those with maximal elements allows us to drop the awkward
condition that the agenda is finite.
Unlike Theorem~\ref{nakamura-thm}, Theorem~\ref{nakamura-max} gives 
a condition for the general nonemptyness of the core for \emph{infinite}, as well as finite, agenda.
Fourth, it fits the Alternative Scenario.  It gives a condition for the planner to be assured of
the existence of a chosen alternative as soon as she presents an agenda to the players
(i.e., before she learns their preferences, supposing that they have maximal elements).
Fifth,  our framework is very general.  Like \citet{nakamura79},
we impose no conditions such as monotonicity or properness on simple games.
Unlike  \citet{nakamura79}, we consider arbitrary sets of players and 
\emph{arbitrary algebras} of coalitions.\footnote{
Most works in this literature consider finite sets of players.
\citet{nakamura79} considers arbitrary (possibly infinite) sets of players and the algebra of all subsets of players.}

\emph{The second main result}, Theorem~\ref{nakamura-max-extended}, is about the
 \emph{kappa number} (Definition~\ref{kappa}), an extension of the Nakamura number
 to the even more general framework that distinguishes 
 the collection~$\B'$ of the sets of players for which one can assign winning/losing status
 from the algebra $\B$ of (identifiable) coalitions.
The kappa number $\kappa(\W')$ is defined for a collection $\W'\subseteq \B'$ of winning sets.
The result asserts that  the following two statements are equivalent:
(i)~the number of alternatives in the agenda~$B$ is less than the kappa number~$\kappa(\W')$;
(ii)~the core $C^+(\W',B,\p)$ without majority dissatisfaction is nonempty for all profiles $\p$ for the agenda.
It also asserts that the above two statements imply, but are not implied by:
(iii)~the core $C(\W',B,\p)$ is nonempty for all profiles $\p$ for the agenda.

Theorem~\ref{nakamura-max-extended} gives another sense in which
the core without majority dissatisfaction behaves better than the core.
Computing the kappa number is not an easy task in general.
So, in applying the theorem, Lemma~\ref{nu-kappa-nu} is useful;
it estimates the kappa number from above and below in terms of the Nakamura numbers.

\section{Preliminaries}

\subsection{Some facts about ordinal numbers}

The notion of \emph{ordinal numbers} (or \emph{ordinals}) generalizes that of natural numbers.
This section presents some facts about ordinals.\footnote{
The paper does not require much knowledge of the theory of ordinal numbers.
Understanding a few notions (such as \emph{limit ordinals} and \emph{cardinal numbers}) and facts
will suffice to understand details of the paper (mostly in footnotes).
A deeper application to economic theory can be found in \citet{lipman91}, 
who uses this theory to find a fixed point of an ``infinite regress'' that a modeler faces.}
Consult a textbook for axiomatic set theory (e.g., \citet{hrbacek-j84})
 for more systematic treatment.

We start by introducing the first ``few'' ordinals.
The natural numbers are constructed from $\emptyset$ as follows:
$0=\emptyset$, 
$1=0\cup \{0\}=\{0\}=\{\emptyset\}$,
$2=1\cup \{1\}=\{0,1\}=\{\emptyset,\{\emptyset\}\}$,
$3=2\cup\{2\}=\{0,1,2\}$, 
$4=3\cup\{3\}=\{0,1,2,3\}$, etc.
The first ordinal number that is not a natural number is the set 
$\omega=\{0,1,2,3,\ldots\}$ of natural numbers.
We can continue the process to obtain
$\omega+1=\omega\cup \{\omega\}=\{0,1,2,\ldots,\omega\}$,
$\omega+2=(\omega+1)+1=(\omega+1)\cup \{\omega+1\}=\{0,1,2,\ldots,\omega,\omega+1\}$, etc.
We then have
$\omega\cdot 2=\omega+\omega=\{0,1,2,\ldots,\omega,\omega+1,\omega+2,\ldots\}$,
$\omega\cdot 2+1$, \ldots, $\omega\cdot 3$, \ldots,
 $\omega\cdot \omega=\{0,1,2,\ldots,\omega,\omega+1,\ldots,\omega\cdot 2,\omega\cdot 2+1,\ldots,
 \omega\cdot 3,\ldots,\omega\cdot 4,\ldots\}$.

For a given ordinal $\alpha$, its \emph{successor} $\alpha+1=\alpha\cup \{\alpha\}$ always exists and is an ordinal.
For a given ordinal $\alpha$,  $\alpha-1$ does not necessarily exist: if there is an ordinal $\beta$ such that 
$\alpha=\beta+1$, then $\alpha$ is a \emph{successor ordinal}; otherwise, it is a \emph{limit ordinal}.
Every natural number except~$0$ is a successor ordinal.
Both $\omega$ and $\omega+\omega$ are limit ordinals.
But $\omega+1$, $\omega+2$, etc.\ are successor ordinals.

Define $\alpha<\beta$ if and only if $\alpha\in\beta$.
$<$ has all the properties of a linear order.
Every set~$A$ of ordinal numbers are \emph{well-ordered} by~$<$,
that is, every nonempty subset of $A$ has a $<$-least element.

Each ordinal $\alpha$ has the property that 
\[
\alpha=\{ \beta: \textrm{$\beta$ is an ordinal and $\beta<\alpha$} \}.
\]
If $\alpha$ is a successor ordinal, say $\beta+1$, then as a set, it has a greatest element, namely~$\beta$.
If $\alpha$ is a limit ordinal, then it does not have a greatest element, and $\alpha=\sup \{\beta: \beta<\alpha \}$.

A function whose domain is an ordinal $\alpha$ is called a \emph{(transfinite) sequence} of length $\alpha$.

Two sets $Y$ and $Y'$ are \emph{equipotent} if there is a bijection (one-to-one and onto function) 
from $Y$ to $Y'$.
An ordinal number $\alpha$ is an \emph{initial} ordinal if it is not equipotent to any $\beta< \alpha$.
For example, $\omega$ is an initial ordinal, because it is not equipotent to any natural number.
$\omega+1$ not initial, because it is equipotent to $\omega$.  Similarly, none of countable ordinals 
$\omega+2$, $\omega+3$, $\omega+\omega$, $\omega\cdot\omega$, $\omega^\omega$, \ldots is initial.

The \emph{cardinal number} of a set $Y$, denoted $\#Y$, is the unique initial ordinal equipotent to $Y$.
In particular, if $Y$ is countable, then $\#Y=\omega$.
There are arbitrarily large cardinal numbers.
Infinite cardinal numbers form a transfinite sequence of \emph{alephs} $\aleph_\alpha$, with $\alpha$ ranging
over all ordinal numbers.
We have $\aleph_\alpha+\aleph_\beta=\aleph_\alpha\cdot \aleph_\beta=\max\{\aleph_\alpha,\aleph_\beta\}$.
Appendix~\ref{card-lim} gives a proof of the following:

\begin{lemma} \label{lem:card-lim}
An infinite cardinal number is a limit ordinal.
\end{lemma}

Without defining the \emph{ordinal sum} and the \emph{cardinal sum} here,
 let us just mention the following useful lemma (proved in Appendix~\ref{sumsum}):
 
\begin{lemma} \label{lem:sumsum}
$\#(\alpha+\beta)=\#\alpha+\#\beta$, where the sum on the left side is the ordinal sum and
the sum on the right is the cardinal sum.
\end{lemma}

\subsection{Framework}

Let $N$ be an arbitrary nonempty set of players and 
$\B\subseteq 2^{N}$ an arbitrary Boolean algebra of subsets of~$N$
(so $\B$ includes $N$ and is closed under union, intersection, and complementation).  
The elements of $\B$ are called \textbf{coalitions}.
Intuitively, they are the observable or identifiable or describable subsets of players.
A ($\B$)-\textbf{simple game} $\W$ is a subcollection of~$\B$ such that $\emptyset\notin \W$ and $\W\neq\emptyset$.
The elements of $\W$ are said to be \textbf{winning}, and the other elements in~$\B$ 
are \textbf{losing}.  
A simple game~$\W$ is \textbf{weak} if the intersection $\cap\W=\cap_{S\in\W}S$
of the winning coalitions is nonempty.

Let $X$ be a (finite or infinite) set of \textbf{alternatives}, with cardinal number $\#X\geq 2$.
\emph{In this paper, a \textbf{(strict) preference} is an \emph{asymmetric} relation $\succ$ on $X$}:
if $x\succ y$ (``$x$ is preferred to~$y$''), then $y\not\succ x$.
A relation~$\succ$ is \textbf{total} if $x\neq y$ implies $x\succ y$ or $y\succ x$.
An asymmetric relation is a \textbf{linear order} if it is transitive and total.
A binary relation $\succ$ on $X$ is \emph{acyclic} if for any finite set $\{x_1, x_2, \ldots, x_m\}\subseteq X$,
whenever $x_1 \succ x_2$, \ldots, $x_{m-1} \succ x_m$, we have $x_m \not\succ x_1$.
Acyclic relations are preferences since they are asymmetric (and irreflexive).
Let $\A$ be the set of \emph{acyclic} preferences on~$X$.

A \textbf{profile} is a list $\p=\profile$ of individual preferences~$\pref$.
Intuitively, $x \pref y$ means that player~$i$ prefers $x$ to~$y$ at profile~$\p$.
A profile~$\p$ is \textbf{($\B$)-measurable} if $\xprefy \in\B$ for all~$x$, $y\in X$.  
Denote by $\acyclicprofs$ the set of all measurable profiles of \emph{acyclic} preferences. 

An \textbf{agenda} (or  ``budget set'' or ``opportunity set'') $B$ is a subset of $X$.
Note that a preference, when restricted to the elements in $B$,  defines an asymmetric relation on~$B$.
Let $B\subseteq X$ be an agenda.
An alternative $x\in B$ is said to be a \textbf{maximal} element of $B$ with respect to a binary relation~$\succ$
(written $x\in \max_B \succ$ though $B$ is often dropped)
if there does not exist an alternative $y\in B$ such that $y\succ x$.
Let $\M(B)$ be the set of \textbf{preferences for~$B$}, i.e., 
\emph{asymmetric} relations~$\succ$ on $X$ 
\emph{that has a maximal element of~$B$}.\footnote{We define a preference for $B$ on $X$,
despite the fact that the set of maximal elements of $B$ with respect to $\succ$ depends only on 
the restricted relation $\succ\cap (B\times B)$.}
Let $\profs$ be the set of \textbf{profiles for~$B$}, i.e.,
measurable profiles $\p=\profile \in\M(B)^{N}$ of preferences $\pref$ for $B$.

A \textbf{($\B$)-simple game with (ordinal) preferences} is a list $(\W, B, \p)$ of
a $\B$-simple game~$\W\subseteq\B$, a subset $B$ of alternatives, and
a profile $\p=\profile$.
Given the simple game with preferences, 
we define the (not necessarily asymmetric) \textbf{dominance} relation~$\dom$ on~$X$
 by $x\dom y$ if and only if there is a winning coalition
$S\in\W$ such that $x\pref y$ for all $i\in S$.\footnote{
In this definition, $\xprefy$ need not be winning since we do not assume $\W$ is monotonic.
\citet{andjiga-m00} study Nakamura's theorem,
adopting the notion of dominance that requires the above coalition to be winning.
Their dominance relation distinguishes the game from its monotonic cover,
while the classical dominance~$\dom$ does not.
The two notions are equivalent if and only if the game is monotonic.} %
The \textbf{core} $C(\W, B, \p)$ of the simple game with preferences is 
the set $\max_B\dom$ of undominated alternatives:
\[
C(\W, B, \p)=\{x\in B: \textrm{$\not\!\exists y\in B$ such that $y\dom x$}\}.
\]

An alternative $x\in B$ is \textbf{Pareto in~$B$} if there exists no $y\in B$ 
such that $\yprefx=N$.
It is easy to prove that any alternative in $C(\W, B, \p)$ or in $\bigcup_i\max\pref$ is Pareto in~$B$.

A \textbf{(preference) aggregation rule}  is
a map $\succ \colon \p\mapsto \, \succ^\p$ from profiles~$\p$ of preferences 
to binary relations (\emph{social preferences})~$\succ^\p$ 
on the set $X$ of alternatives.  For example, the mapping $\succ_\W$ 
from profiles $\p\in \acyclicprofs$ of acyclic preferences to dominance relations $\dom$ 
is an aggregation rule.

\subsection{Nakamura's theorem for acyclic preferences}

\citet{nakamura79} gives a condition for a $2^{N}$-simple game with preferences
to have a nonempty core for any profile~$\p$ of \emph{acyclic} preferences.
To state Nakamura's theorem, we define the \textbf{Nakamura number} $\nu(\W)$ of a $\B$-simple game~$\W$
to be the size of the smallest collection of winning coalitions having empty intersection\footnote{The minimum
of the following set of cardinal numbers exists since every set of ordinal numbers has a
$<$-least element.}
\[
\nu(\W)=\min \{\#\W': \textrm{$\W' \subseteq \W $ and $\bigcap\W'=\emptyset$} \}
\]
if $\bigcap\W:=\bigcap_{S\in\W}S=\emptyset$ (i.e., if $\W$ is \emph{nonweak}); 
otherwise, set $\nu(\W)=+\infty$, which is understood to be greater than any cardinal number.
By the assumption that $\emptyset\notin \W$ and $\W\neq\emptyset$, we have $\nu(\W)\ge 2$.
It is easy to prove the following lemma:\footnote{This result can be found in \citet[Lemma~2.1 and Corollary~2.2]{nakamura79}.}

\begin{lemma} \label{nakamura-bounds}
If $\W$ is a nonweak $\B$-simple game, then 
$2\le \nu(\W)\leq \min \{\#S: S\in\W\}+1$ and $\nu(\W)\le \#N$.\end{lemma}

The following theorem extends Nakamura's result~\citep{nakamura79} for $\B=2^{N}$:

\begin{theorem}[\protect \citet{kumabe-m08jme}] \label{nakamura-thm}
Let $\W\subseteq\B$ be a simple game and $B\subseteq X$ an agenda.
Then the core $C(\W, B, \p)$ is nonempty 
for all measurable profiles $\p\in \acyclicprofs$ of \emph{acyclic} preferences
if and only if \emph{$B$ is finite} and $\#B<\nu(\W)$.
\end{theorem}

Since $\dom$ is acyclic if and only if the set $C(\W,B,\p)$ of maximal elements of~$B$ with
respect to~$\dom$ is nonempty for every finite subset $B$ of $X$, we have:

\begin{cor} \label{nakamura-cor}
The dominance relation~$\dom$ is acyclic for all $\p\in \acyclicprofs$ if and only if 
$\#B<\nu(\W)$ for all \emph{finite} $B\subseteq X$. 
\end{cor}

\section{Two notions of the core}

In this section, we first introduce the notion of the \emph{core without majority dissatisfaction},
 a strengthening of the core.
We then compare the two notions of the core, focusing on how they treat 
 the maximal elements of individual preferences.

We consider $\B$-simple games $(\W, B, \p)$ with preferences, given for each profile~$\p$ (not necessarily in $\profs$).
An alternative $x\in B$ is \emph{not} in the core $C(\W, B, \p)$ if $x$ is not maximal with respect to the
dominance relation $\dom$: there are $y\in B$ and a winning coalition $S\in\W$
such that for all~$i\in S$, $y_i=y$ and $y_i\pref x$.
(That is, $x\in B$ is not in $C(\W,B,\p)$ if for some $y\in B$, the set $\{i: y\pref x\}$ contains a winning coalition.)
So, an alternative $x$ (e.g., $d$ in Examples \ref{ex:core1} and \ref{ex:core2} below) can be in the core even if
there is a winning coalition contained in
the set of players $i$ that prefer some $y_i$ to $x$, as long as $y_i$ is different among the
players.  To exclude such an $x$ from the core, we modify the definition:

\begin{definition} \label{c+}
An alternative $x\in B$ is in the \textbf{core $C^+(\W,B,\p)$ without majority dissatisfaction}
if there is no winning coalition $S\in\W$ such that for all $i\in S$, there exists some~$y_i\in B$ satisfying $y_i\pref x$.\footnote{%
To belong to $C^+(\W,B,\p)$, an alternative must be a maximal element for at least
one individual in each $S\in\W$.
So we can rewrite $C^+(\W,B,\p) = \bigcap_{S\in\W} \bigcup_{i\in S} \max_B \pref$.}
In other words, $x\in B$ is in $C^+(\W,B,\p)$ if the set $\{i: x\notin \max_B \pref\}=\{i: y\pref x$ for some $y\in B\}$ 
of players for whom $x$ is non-maximal (players ``dissatisfied with $x$'') contains no winning coalition.\footnote{\label{foot:core_scheme}
The core without majority dissatisfaction consists of those alternatives not rejected 
by the following scheme (when the players behave sincerely):
After proposing an agenda~$B$, for each alternative $x\in B$, 
the planner does the following:
(i)~she proposes~$x$ to the players;
(ii)~she asks each player $i$ whether $i$ is ``dissatisfied'' with~$x$ 
(which, by definition, means whether $i$~finds some alternative in~$B$ better than~$x$),
without asking what alternative $y_i$ is better;
(iii)~if a winning coalition of players are ``dissatisfied'' with~$x$, then the planner rejects~$x$.
Under this scheme, the planner elicits individual preferences in a very incomplete manner
(as is usual with real-world decisions).
Also, the members of the winning coalition are only united in their opposition to~$x$:
they do not have to agree on an alternative~$y$ that should replace~$x$;
they do not even have to know what alternatives $y_j$ the other members $j$~prefer.
See Remark~\ref{remark:core_story} for further discussion.}\end{definition}

\begin{remark} \label{remark:core}
The word ``core'' usually refers to the set of maximal alternatives with respect to some relation.
We adopt the word since 
the core without majority dissatisfaction is indeed the set of maximal elements with respect to the following 
extended dominance relation~$\dom$, which relates a subset~$Y$ of alternatives to 
an alternative~$x$.\footnote{
The extended dominance~$\dom$ can be seen as a dominance relation of 
a \emph{game in constitutional form} \citep{andjiga-m89}, 
which associates a simple game with each pair of subsets of alternatives.
Note, however, that the dominance relations (e.g., \emph{$i$-domination}) they actually analyze, unlike ours,
require each player in a locally winning coalition to prefer \emph{all} the alternatives in $Y$ to~$x$.} %
First, we extend $i$'s preference $\pref\subseteq X\times X$ to a relation $\pref\subseteq 2^X\times X$
 (where $2^X$ is the power set of $X$):
$Y\pref x$ if and only if there is $y\in Y$ such that $y\pref x$.\footnote{
We are following the consequentialist approach of extending preferences on~$X$ to ones on its power set---the set
of opportunity sets.
Unless one is concerned with preferences for flexibility~\citep[e.g.,][]{kreps79} or freedom of choice,
this approach is standard.}
Next, we extend the dominance relation $\dom\subseteq X\times X$ to a relation $\dom\subseteq 2^X\times X$:
$Y\dom x$ if and only if there is $S\in \W$ such that for all $i\in S$, $Y\pref x$.
Then, $x\in C^+(\W,B,\p)$ if and only if there is no $Y\subseteq B$ such that $Y\dom x$
(if and only if $B\not\dom x$).\end{remark}

\begin{remark}\label{remark:core_story}
The core without majority dissatisfaction rejects any alternative (``status quo'') $x$ dominated
by some set~$Y$ of alternatives with respect to the extended dominance relation in Remark~\ref{remark:core}.
According to this dominance relation, a coalition can block~$x$
\emph{without having to agree on a replacement alternative}.
Admittedly, this notion of dominance may lack strong support, especially if 
one sticks to the usual interpretation of alternatives as complete descriptions of social states.
However, when the standard solution (the core) selects too many alternatives,
deleting some of them on a relatively weak ground could normatively be desirable.\footnote{After all, 
the dominance relation that defines the \emph{core} has relatively weak support 
in view of the \emph{stability set} proposed by Rubinstein, 
for example:
a coalition rejects alternatives without taking into consideration that the dominating alternative
may further be rejected.}
The point of the main results is that they provide a condition (Nakamura's inequality)
ensuring that something remains even after alternatives are rejected on weak grounds.

Having said that, we give an example where 
blocking behavior without agreeing on a single alternative~$y$ is plausible.
One such example is a \emph{popularity contest among certain goods}.
Consider the problem of selecting the ``best'' articles published in economic journals in 2010, for instance.
For simplicity, assume that $N=\{1,2,3\}$---the opinions (preferences) of only three experts are elicited.
Since one wants to compare different articles, a natural candidate
for an alternative is an \emph{article}, which is not a social state.
Then the idea in Remark~\ref{remark:core} of a set~$Y$ dominating an article~$x$
should be intuitive enough, since it can be \emph{restated} as follows:
there are a majority, say the coalition $\{1,2\}$, of experts~$i$
and a feasible social state $(y_1, y_2,x)\in Y^3$ that dominates the feasible social state $(x,x,x)$
in the conventional sense:
each $i\in\{1,2\}$ prefers $(y_1, y_2,x)$ to $(x,x,x)$; that is, $1$ prefers $y_1$ to $x$
and $2$ prefers $y_2$ to $x$.
When one can assume that each allocation $(y_1,y_2,y_3)\in Y^3$ is feasible, 
this restatement makes a perfect sense.\footnote{
\emph{If each $y\in Y$ is a disposable private good} that is available in sufficient quantity,
the assumption is satisfied.  
Note that we require a majority despite dealing with a private good, since we are considering a popularity contest.
 \emph{If each $y\in Y$ is a public good}, $(y_1,y_2,y_3)$ generally describes an \emph{infeasible}, imaginary social state
in which each $i$ consumes the public good~$y_i$.
The restatement loses some validity because of the lack of feasibility.
Nevertheless, it does not lose all the validity in our view, 
since preferences are often elicited in a very incomplete manner
 in the real world, as footnote~\ref{foot:core_scheme} suggests.} %
Anyone accepting the conventional dominance relation defining the core
would be ready to accept the extended dominance relation.\end{remark}

The following two lemmas are immediate from the definition.  
The first one says that the core without majority dissatisfaction depends only on the set of maximal elements
of each player.

\begin{lemma} \label{c+maximal}
Let $\p$, $\p'$ be two profiles satisfying $\max_B \pref=\max_B \succ^{\p'}_i$ for all~$i$.
Then $C^+(\W,B,\p)=C^+(\W,B,\p')$.\end{lemma}

\begin{lemma} \label{c+in_c}
For each profile $\p$, we have $C^+(\W,B,\p)\subseteq C(\W,B,\p)$.
The inclusion $\subseteq$ is strict for some profile. \end{lemma}

Each of Examples \ref{ex:core1} and \ref{ex:core2} below\footnote{See Appendix~\ref{appendix-figs} for graph representations of the profiles in these examples.}
 shows that the inclusion is sometimes strict.
Example \ref{ex:core1} also shows that \emph{an alternative can be in the core even if it
is not maximal with respect to anyone's preference}:

\begin{example} \label{ex:core1}
Let $N=\{1,2, 3\}$ and let $\W$ consist of the coalitions having a majority (i.e., having at least two players).
Let $X=\{a,b,c,d,e\}$.  Define a profile by 
$\succ_1=\{(a,d),(e,b),(e,c)\}$,
$\succ_2=\{(b,d),(e,a),(e,c)\}$, and
$\succ_3=\{(c,d),(e,a),(e,b)\}$.
Then the sets $\max \succ_i$ of maximal elements of $X$ with respect to $\succ_i$
 are given by
$\max \succ_1=\{a,e\}$,
$\max \succ_2=\{b,e\}$, and
$\max \succ_3=\{c,e\}$.
But the core $C$ is $\{d, e\}$.
So, the core is neither a subset nor a superset of $\bigcup_i \max\succ_i=\{a, b, c, e\}$.
The core $C^+$ without majority dissatisfaction is $\{e\}$, 
which is a subset of $\bigcup_i \max\succ_i$.\end{example}

Example~\ref{ex:core2} also shows that \emph{the core, even if it is nonempty,
does not necessarily intersect the union of the maximal elements of individual preferences}:

\begin{example} \label{ex:core2}
We modify the simplest voting paradox (a cycle involving three alternatives and three players)
 by adding an alternative~$d$.
Let $N=\{1,2, 3\}$ and let $\W$ consist of the coalitions having a majority.
Let $X=\{a,b,c,d\}$.  Define a profile by 
$\succ_1=\{(a,b),(b,c),(a,d)\}$,
$\succ_2=\{(b,c),(c,a),(b,d)\}$, and
$\succ_3=\{(c,a),(a,b),(c,d)\}$.
Then the core $C$ is $\{d\}$.
So it does not intersect $\bigcup_i \max\succ_i=\{a, b, c\}$.
The core $C^+$ without majority dissatisfaction is empty.
\end{example}

Unlike the core, the core without majority dissatisfaction is always included
 in the union of the maximal elements of individual preferences:

\begin{lemma} \label{c+in_maximals}
For each profile $\p$,  we have 
$C^+(\W,B,\p)\subseteq C(\W,B,\p)\cap(\bigcup_i \max_B \pref)$.\footnote{
Appendix~\ref{ex:core3} shows that the inclusion is strict for some profile.}\end{lemma}

\begin{proof}
By Lemma~\ref{c+in_c}, it suffices to show that $C^+$ is a subset of $\bigcup_i \max_B \pref$.
Suppose $x\in B$ but $x\notin \bigcup_i \max_B \pref$.
Then, $x\notin  \max_B \pref$ for any~$i\in N$.
This implies that $\{i: x\notin \max_B \pref\}=N\supseteq S$ for any winning coalition $S\in \W$.
(Such an $S$ exists since $\W\neq\emptyset$.)
By the definition of $C^+$, we have $x\notin C^+$.\end{proof}

\section{Preferences with Maximal Alternatives}

We consider in the rest of the paper profiles
\emph{for} a set~$B$ of alternatives, that is, measurable profiles consisting
of \emph{preferences that have a maximal element of~$B$}.

\subsection{The results for the core of a simple game}

We now give a version of Nakamura's theorem for profiles \emph{for} a set~$B$ of alternatives:

\begin{theorem} \label{nakamura-max}
Let $\W\subseteq\B$ be a simple game and $B\subseteq X$ an agenda.
Let $\profs$ be the set of measurable profiles of preferences having a maximal element of~$B$.
Then the following statements are equivalent:\footnote{
The implication (iii)$\Rightarrow$(ii) is \emph{not} the same as the following statement (Example~\ref{ex:core2}):
for each $\p\in \profs$,
if $C(\W, B, \p)\neq\emptyset$, then $C^+(\W, B, \p)\neq\emptyset$.}
\\
\textup{(i)}~$\#B<\nu(\W)$;\\
\textup{(ii)}~the core $C^+(\W, B, \p)$ without majority dissatisfaction is nonempty for  all $\p\in \profs$;\\
\textup{(iii)}~the core $C(\W, B, \p)$ is nonempty for all $\p\in \profs$.
\end{theorem}

\begin{proof}
(i)$\Rightarrow$(ii). 
Suppose $C^+(\W, B, \p)=\emptyset$ for some profile $\p$ for~$B$.
By Definition~\ref{c+}, for each $x\in B$, there is a winning coalition $S_x\in\W$
such that $S_x\subseteq \{i: x\notin\max\pref\}$.
We claim that $\bigcap_{x \in B}S_x=\emptyset$.
(Otherwise, there is an $i$ who is in $S_x$ for all $x\in B$.
It follows that $\pref$ has no maximal element of~$B$.)
The claim shows that $\nu(\W)\leq \#B$.

\medskip
(ii)$\Rightarrow$(iii).  Immediate from Lemma~\ref{c+in_c}.

\medskip
(iii)$\Rightarrow$(i). 
Suppose $\# B\geq \nu(\W)$.  We construct a profile $\p$ for~$B$ such that $C(\W, B, \p)=\emptyset$.
Let $\nu=\nu(\W)\ge 2$.

\emph{Step}~1, \emph{Case}~(a): $\nu$ is finite. 
\emph{We construct a profile $\p$ such that the dominance relation~$\dom$ has a cycle} consisting of
$\nu$ alternatives
(and the cycle contains an alternative $x_0$ greater than any alternative $y$ not belonging to the cycle).

By the definition of the Nakamura number, there is a collection 
$\W'=\{L_0, \ldots, L_{\nu-1} \}$ of winning coalitions such that
$\bigcap\W'=\bigcap_{k=0}^{\nu-1} L_k=\emptyset$.

Choose a subset $B'=\{x_0, x_1, \ldots, x_{\nu-1} \} \subseteq B$ and write $x_\nu=x_0$.
Fix a cycle (noting that $(x,y)\in\,\succ$ means $x\succ y$)
\[
\succ\, = \{(x_{k+1},x_k): k\in \{0, \ldots, \nu-1 \}\}
\]
and a relation $\succ'=\{(x_0,y): y\notin B' \}$.
Now, go to Step~2.

\emph{Step}~1, \emph{Case}~(b): $\nu$ is infinite.
\emph{We construct a profile $\p$ such that the dominance relation~$\dom$ defines an increasing
 transfinite sequence} of length $\nu$ 
 (and the sequence contains an alternative $x_0$ greater than any alternative $y$ not belonging to the sequence).
 
Recall that if $\nu$ is an ordinal, then $\nu=\{ \alpha: \textrm{$\alpha$ is an ordinal and $\alpha<\nu$} \}$.
By the definition of the Nakamura number, there is a collection\footnote{\label{note:indexed-collection}
Since $\W'$ is a (well-orderable) set whose cardinal number is $\nu$,
there is a bijection that maps each $\alpha\in\nu$ into an element $L_\alpha$ of $\W'$.
So we can write $\W'=\{L_\alpha: \alpha \in \nu \}$.} %
$\W'=\{L_\alpha: \alpha\in\nu \}$ of winning coalitions $L_\alpha$ 
such that $\bigcap\W'=\bigcap_{\alpha<\nu} L_\alpha=\emptyset$. 

Choose a subset $B'=\{x_\alpha : \alpha\in \nu \} \subseteq B$.
Fix a relation\[
\succ\, =  \{(x_{\alpha+1},x_\alpha): \alpha\in \nu \},
\]
which defines an increasing transfinite sequence of alternatives,\footnote{%
The sequence $(x_\alpha)_{\alpha\in\nu}$ does not end: if $\alpha\in\nu$, then $\alpha+1\in\nu$.
This is because $\nu$, being a Nakamura number, is a cardinal number and 
any infinite cardinal number is a limit ordinal by Lemma~\ref{lem:card-lim}.}
and another relation $\succ'=\{(x_0,y): y\notin B' \}$.

\medskip

\emph{Step}~2.
We define a profile $\p=\profile$ by specifying, for each pair $(x,y)\in X^2$,
 the set $\xprefy$ of players that prefer $x$ to~$y$.
Note that $\succ\cap\succ'=\emptyset$.
 (In the following, the letter $\alpha$ denotes an ordinal number,
including a natural number denoted~$k$ in Case~(a) above.)
If $(x,y)=(x_{\alpha+1},x_\alpha)\in \,\succ$, let $\xprefy=L_\alpha$.
If $(x,y)\in \,\succ'$, let $\xprefy=N$.
If $(x,y)\notin \, \succ\cup \succ'$, let $\xprefy=\emptyset$.
The profile~$\p$ is measurable since $L_\alpha$, $N$, $\emptyset\in\B$.

The profile~$\p$ is for~$B$ (i.e., $\p\in\profs$), since we can show that 
each $i$'s preference $\pref$ has a maximal element of~$B$.
For example, if $i\in L_2\cap L_3\cap L_5$,
but $i\notin L_k$ for $k\notin\{2,3,5\}$, then every alternative in $B'$ except 
$x_2$, $x_3$, $x_5$ is a maximal element of $\pref$.
(More formally, the set of maximal elements of $\pref$ is 
$\{x_\alpha\in B': i \notin L_\alpha \}$, which is
nonempty since $\bigcup L_\alpha^c=N$, where $L_\alpha^c=N\setminus L_\alpha$.)

The dominance relation~$\dom$ is clearly $\succ \cup \succ'$,
 since $L_\alpha$ is winning, $N$ contains a winning coalition, and $\emptyset$ is losing.
It follows that the core $C(\W, B, \p)$ is empty.\end{proof} 

\begin{remark}
It is instructive to compare the proof of the ``$\Leftarrow$'' direction of Theorem~\ref{nakamura-thm}
with a direct proof of (i)$\Rightarrow$(iii), which can be given in a way similar to that of (i)$\Rightarrow$(ii) above.
These proofs go as follows:
Suppose $C^+(\W, B, \p)=\emptyset$ for some profile $\p$.
For ``$\Leftarrow$'' of Theorem~\ref{nakamura-thm}, $\p\in \acyclicprofs$
 is a profile of \emph{acyclic} preferences; 
 for (i)$\Rightarrow$(iii), $\p\in \profs$ is a profile of preferences having a maximal element (of~$B$).
Then,\\
(a)~for each $x\in B$, there are $y\in B$ and a winning coalition $S_x\in\W$
such that $S_x\subseteq \{i: y\pref x\}$ (hence $y\dom x$).\\
Note that $x$ is not a maximal element for the players in~$S_x$.
The desired conclusion  $\nu(\W)\leq \#B$ follows if we show $\bigcap_{x \in B}S_x=\emptyset$.
So suppose there is an $i$ who is in $S_x$ for all $x\in B$.
\begin{itemize}
\item To show ``$\Leftarrow$'' of Theorem~\ref{nakamura-thm}, assume $B$ is finite.
By repeated application of~(a), we find the dominance~$\prec_\W^\p$ contains an infinitely ascending sequence:
$x_0\prec x_1\prec x_2\prec x_3\prec \cdots$.  Since $B$ is finite, the sequence contains a cycle such as 
$x_2\prec x_3\prec x_4\prec x_2\prec x_3\prec \cdots$.
It follows that $i$'s preference contains the same cycle, which violates the assumption.

\item For (i)$\Rightarrow$(iii) of Theorem~\ref{nakamura-max}, a contradiction is immediate:
let $x\in B$ be a maximal element for~$\pref$; then $i\in S_x$ is violated.
We do not even need the assumption that $B$ is finite.
We remark that the possibility remains that $i\in S_x$ if $x\in B$ is not a maximal element for~$i$.
So $i$'s preference may contain a cycle without contradiction; it consists only of non-maximal elements for her.
\end{itemize}
\end{remark}

The following example is an application of Theorem~\ref{nakamura-max}.
It gives a condition for an \emph{infinite} set of alternatives to have an element in the core, that is,
 a maximal element with respect to the dominance relation~$\dom$.

\begin{example}
Let $N=[0,1]$ be the unit interval on the set of real numbers and let
$\B$ consist of all subsets of~$N$.
Define a simple game $\W$ by $S\in\W$ if and only if $S^c$ is countable.
Then, it is easy to show that $\nu(\W)$ is uncountable.
Let $X$ be a countable set of alternatives (e.g., the set of rational numbers in $[0,1]$).
Suppose that for each $i$, her preference $\pref$ has a maximal alternative
(e.g., a utility function representing~$\pref$ has finitely many ``peaks,'' all corresponding to rational numbers).
Then, Theorem~\ref{nakamura-max} implies that $C^+(\W,X,\p)$ and $C(\W,X,\p)$ 
both contain alternatives.\end{example}

The profiles constructed in the proof of (iii)$\Rightarrow$(i) of Theorem~\ref{nakamura-max}
consist of individual preferences that may have more than one maximal alternative.
However, we can modify the proof in such a way that each individual preference has 
\emph{exactly one} maximal alternative.
That gives the following proposition:

\begin{prop} \label{nakamura-max2}
Let $\W\subseteq\B$ be a simple game and $B\subseteq X$ an agenda.
Then the three equivalent statements \textup{(i)}, \textup{(ii)}, and \textup{(iii)} in Theorem~\ref{nakamura-max}
are equivalent to the following statements:\\
\textup{(ii.a)}~The core $C^+(\W, B, \p)$ without majority dissatisfaction is nonempty for all measurable profiles~$\p$ 
of preferences having \emph{exactly one} maximal element of~$B$.\\
\textup{(ii.b)}~The core $C^+(\W, B, \p)$ without majority dissatisfaction is nonempty for all measurable profiles~$\p$ 
of linear orders having a maximal element of~$B$.\\
\textup{(iii.a)}~The core $C(\W, B, \p)$ is nonempty for all measurable profiles~$\p$ of preferences 
having \emph{exactly one} maximal element of~$B$.\\
\textup{(iii.b)}~The core $C(\W, B, \p)$ is nonempty for all measurable profiles~$\p$ of linear orders 
having a maximal element of~$B$.\end{prop}

\begin{proof}
(ii)$\Rightarrow$(ii.a)$\Rightarrow$(ii.b) and (iii)$\Rightarrow$(iii.a)$\Rightarrow$(iii.b) are obvious.
(ii.b)$\Rightarrow$(iii.b) is immediate from Lemma~\ref{c+in_c}.

(iii.b)$\Rightarrow$(i).
Suppose $\# B\geq \nu(\W)$.  We construct a profile $\p$ satisfying the condition
 such that $C(\W, B, \p)=\emptyset$.
Let $\nu=\nu(\W)\ge 2$.
In \emph{Step}~1 of the proof of Theorem~\ref{nakamura-max}, replace the relation~$\succ'$
by an arbitrary asymmetric subrelation $\succ'\, \subset X \times (X\setminus B')$
that is linear on $X\setminus B'$ 
and satisfies $x \succ' y$ whenever $x \in B'$ and $y \in X\setminus B'$.
We replace \emph{Step}~2 of the proof by the following argument.

\emph{Case}: $\nu$ is finite. 
Define $L_{-1}=N$ and 
for all $k\in \{0, \ldots, \nu-1 \}$, 
\[
D_k=(L_{-1}\cap L_0\cap \cdots \cap L_{k-1})\setminus L_k.
\]
Then $\{D_0, \ldots, D_{\nu-1} \}$ is a family of (possibly empty) pairwise disjoint coalitions in~$\B$
such that $L_k\subseteq D_k^c:=N\setminus D_k$ for all $k$ and $\bigcup_{k=0}^{\nu-1}  D_k=N$
($i\in N$ is in the first $D_k$ such that $i\notin L_k$).

Define $\p$ as follows:  for each $k$, all players~$i$ in $D_k$ have the same
linearly ordered preference~$\pref$ with maximal element~$x_k$,
obtained by taking the transitive closure of 
$\succ\setminus\{(x_{k+1},x_k)\}\cup \succ'$.
Obviously, $\p$ is measurable since for each $x$, $y\in X$, 
$\xprefy$ is a finite union of~$D_k$'s.  
Also,  $\dom$ includes $\succ\cup \succ'$.
(If $(x,y)=(x_{k+1},x_k) \in \, \succ$, we have
$\xprefy=D_k^c\supseteq L_k\in\W'$.  
If $(x,y)\in \,\succ'$, then $\xprefy=N$.)
It follows that $C(\W, X, \p)=\emptyset$.

\medskip

\emph{Case}: $\nu$ is infinite.
Note that $(B'\times B')\cap \succ'=\emptyset$.
Define $\p$ as follows:
If $(x,y)\in \,\succ'$, then $\xprefy=N$.
If $(x,y)\notin (B'\times B')\cup \succ'$, then $\xprefy=\emptyset$. 
If $(x,y)=(x_\alpha,x_\beta)\in B'\times B'$, 
$x_\alpha\pref x_\beta$ if and only if either\footnote{
For each $i$, $\{x_\alpha: i\in L_\alpha^c \}\neq \emptyset$ is $i$'s \emph{preferred class} of alternatives in $B'$
and $\{x_\alpha: i\in L_\alpha\}$ is her \emph{less-preferred class}.
Her preference $\pref$ is linear on~$B'$ satisfying the following conditions:
(a)~if two alternatives belong to her preferred class, then $i$ prefers the one with the smaller index;
(b)~$i$ prefers each alternative in her preferred class to any in her less-preferred class;
(c)~if two alternatives belong to her less-preferred class, then $i$ prefers the one with the greater index.
For example, if $i\in L_2\cap L_3\cap L_5$, but $i\notin L_\alpha$ for $\alpha\notin\{2,3,5\}$, 
then $i$ ranks the alternatives in the following order: $x_0$, $x_1$, $x_4$, $x_6$, $x_7$, \ldots, $x_5$,
$x_3$, $x_2$.} %
\begin{itemize}
  \item $i\in L_\alpha^c\cap L_\beta^c$ and $\alpha<\beta$; or 
  \item $i\in L_\alpha^c\cap L_\beta$ and $\alpha\ne \beta$; or 
  \item $i\in L_\alpha\cap L_\beta$ and $\alpha>\beta$.
\end{itemize}
This is equivalent to saying that
\footnote{
$x_\alpha\pref x_\beta$ iff either
(1)~$\alpha<\beta$ and $i\in (L_\alpha^c\cap L_\beta^c)\cup(L_\alpha^c\cap L_\beta)
=L_\alpha^c\cap(L_\beta^c\cup L_\beta)=L_\alpha^c$
or
(2)~$\alpha>\beta$ and $i\in (L_\alpha^c\cap L_\beta)\cup(L_\alpha\cap L_\beta)=L_\beta$.
Note that the first two cases can be restated: 
if $\alpha<\beta$,
then $\{i: x_{\beta}\pref x_\alpha\}=L_\alpha$ and
$\{i: x_\alpha\pref x_{\beta}\}=L_\alpha^c$.} %
\[
\xprefy=\{i:x_\alpha\pref x_\beta\}=
	\left\{ \begin{array}{ll}
	L_\alpha^c    & \mbox{if $\alpha<\beta$}, \\
	L_\beta      & \mbox{if $\alpha>\beta$}, \\
	\emptyset & \mbox{if $\alpha=\beta$}.
	\end{array}
	\right.
\]

The profile $\p$ is measurable since $N$, $\emptyset$, $L_\alpha^c$, $L_\beta\in \B$.

Clearly, each $i$'s preference $\pref$ is a linear order on~$X$ that has a unique maximal element of~$B$,
namely the alternative in her preferred class 
$\{x_\alpha: i\in L_\alpha^c \}\neq \emptyset$ with the smallest index.

The dominance relation $\dom$ clearly includes $\succ\cup \succ'$,
since $\{i: x_{\beta+1}\pref x_\beta \}=L_\beta\in \W$ for each $\beta\in \nu$, for example.
It follows that $C(\W, X, \p)=\emptyset$.\end{proof}

\begin{remark}
The argument for the finite $\nu$ case of the proof does not go through for the infinite case, 
since $D_k$ does not necessarily belong to $\B$ when $k$ is infinite.
The argument for the infinite case causes difficulty when applied to the finite case.
For example, suppose $\nu=4$ and $i \in L_0 \cap L_3$.
Since $i\in L_3$ and $4>3$, we have $x_4\pref x_3$ by the definition of~$\p$.
Since $i\in L_0$ and $3>0$, we have $x_3\pref x_0$.
Since $x_0=x_4$, $\pref$ is not asymmetric.\end{remark}

\subsection{The results for the core of a collection of winning sets}

\subsubsection{Extended framework}

We extend the framework by introducing a collection~$\B'$, consisting of subsets of the set~$N$ of players.
Recall that $\B$~consists of the \emph{coalitions}---intuitively, they are the observable or identifiable sets of players.
In contrast, $\B'$~consists of the sets \emph{for which one can assign winning/losing status}---sometimes 
they are the sets whose size is well-defined (Example~\ref{borel-lebesgue}); 
other times they are the sets that are half-identifiable or listable (Example~\ref{recursive-re}).
We assume $\B\subseteq \B'$, which means that one can assign such a status for any coalition.

A \textbf{collection $\W'$ of winning sets} is a subcollection of~$\B'$ satisfying 
$\emptyset\notin \W'$ and $\W'\neq \emptyset$.
Given $\W'$, the most natural simple game one can define is the following:
the \textbf{$\B$-simple game $\W$ induced by $\W'$} consists of the winning coalitions 
(winning sets that are also coalitions), that is, $\W=\W'\cap\B$.
We can define the \emph{core} and the \emph{core without majority dissatisfaction} as before, 
with $\W$ replaced by $\W'$.  Lemma~\ref{c+in_c} and $\W\subseteq \W'$ imply the following.

\begin{lemma} \label{bet-c+c}
Let $\W'\subseteq\B'$ and $\W=\W'\cap \B$.  Then the following statements are true:\\
\textup{(i)}~$C^+(\W',B,\p)\subseteq C^+(\W,B,\p)\subseteq C(\W,B,\p)$.\\
\textup{(ii)}~$C^+(\W',B,\p)\subseteq C(\W',B,\p)\subseteq C(\W,B,\p)$.
\end{lemma}

Since we had better be able to identify who prefers a given alternative to another, 
we leave the definitions of \emph{measurable} profiles and \emph{profiles for} an agenda unaltered.

\subsubsection{Justification for the extended framework}

We have assumed $\B'=\B$ in the previous sections, 
but there are reasons for distinguishing the collection~$\B'$ 
from the Boolean algebra~$\B$ of coalitions.
We illustrate that by two examples in this section.
We also give an example of a decision-making situation where the framework makes sense.

The following example gives some, though limited, justification for the extended framework:

\begin{example} \label{borel-lebesgue}
Let $N=[0,1]$ be the unit interval on the set of real numbers, 
$\B$ the $\sigma$-algebra of Borel sets (i.e., $\B$ is the smallest $\sigma$-algebra containing all open sets in $[0,1]$),
$\B'$ the $\sigma$-algebra of Lebesgue measurable sets, and
$\mu$ Lebesgue measure.
Let $\W'\subseteq\B'$ be any collection of winning sets defined in terms of the measure alone.
(For example, $\W'$ consists of the sets $S\in\B'$ satisfying $\mu(S)>1/2$.
Alternatively, $\W'$ consists of the sets $S\in\B'$ satisfying $\mu(S)\geq 2/3$.)
We do have $\B'\supsetneq \B$, though this fact is not often exploited.\footnote{
For example, \citet{dasgupta-m08} adopt this framework
to formalize the concept that an axiom is satisfied for almost all profiles.
However, their focus is measurable profiles, which means they are mainly interested in 
certain coalitions in~$\B$.
 \citet{banks-dl06}, when restricted to our framework, focus on simple games~$\W\subseteq\B$, 
 rather than on $\W'\subseteq \B'$.} %
Observe that for each $S'\in \B'$ there is $S\in \B$ such that $S\subseteq S'$ and $\mu(S)=\mu(S')$.\footnote{
Let $S'\in\B'$.  Let $E'=S'^c\in \B'$. Apply the following proposition \citep[page~293]{royden88}
and let $S=E^c$: 
If $E'\subseteq [0,1]$ is any set, then there is a Borel set $E\in \B$ such that $E'\subseteq E$ and $\mu^*(E')=\mu(E)$,
where $\mu^*$ is Lebesgue outer measure.} %
\end{example}

The last observation in Example~\ref{borel-lebesgue} explains why it fails to give a compelling justification for
the extended framework, because of the following lemma.
The lemma suggests that under a certain condition, it is not very meaningful to introduce a collection $\W'\subseteq \B'$,
even when it makes sense to extend $\B$ to $\B'$.
The proof is straightforward.

\begin{lemma}
Let $\W'\subseteq\B'$ and $\W=\W'\cap \B$.
Suppose that for each winning set $S'\in \W'$, 
there exists a winning coalition $S\in\W$ satisfying $S\subseteq S'$.
Then $C(\W,B,\p)=C(\W',B,\p)$ and $C^+(\W,B,\p)=C^+(\W',B,\p)$.\end{lemma}

Now we give a more compelling justification for the framework in which $\B'$ and $\W'$
are introduced:
\begin{example} \label{recursive-re}
Let $N=\{0,1,2,\ldots\}$ be the set of natural numbers.
A natural and faithful way to model \emph{identifiable} or \emph{half-identifiable} sets of players is to let
$\B$ be the algebra of recursive sets (coalitions) and
$\B'$ the lattice of r.e.\ sets.\footnote{
According to \emph{Church's thesis}~\citep{odifreddi92}, a set of players (natural numbers)
 is \emph{recursive} if
there is an algorithm that, given any player, will decide whether she is in the set.
A set of players is \emph{r.e.\ (recursively enumerable)} if there is an algorithm that lists, in some order,
the members of the set; the condition in general does not mean that there is an algorithm to decide 
whether a given player belongs to it.
A set $A\subseteq N$ is recursive if and only if $A$ and $A^c$ are both r.e.
\citet{odifreddi92} gives detailed discussion of recursion theory (computability theory).
The papers by \citet{mihara97et,mihara99jme} contain short reviews of recursion theory.} 

The \emph{first reason} for introducing~$\B'\supseteq \B$ in our framework is that we cannot let $\B$ be
the lattice of r.e.\ sets, since the lattice is not a Boolean algebra.

The \emph{second reason} is that the natural system $(W_e)_{e\in N}$ for indexing
r.e.\ sets is easier to handle than
the system $(\varphi_e)_{e\in N}$ that can be used for indexing recursive sets
(where $\varphi_e$ is the $e$th partial recursive function and $W_e$ its domain).
For example, while $\varphi_e$ corresponds to (i.e., is the characteristic function of) 
no recursive set for some $e\in N$
(and there is no algorithm to decide whether a given $e$ corresponds to some recursive set),
$W_e$ corresponds to (i.e., is) an r.e.\ set for any $e\in N$.
One can therefore write any class of r.e.\ sets as $\{W_e: e\in I\}$ for some 
(not necessarily r.e.) set~$I$ of indices,
without worrying that some $W_{e}$ might correspond to no r.e.\ set.
\citet[page~226]{odifreddi92} gives more reasons for preferring $(W_e)_{e\in N}$
to $(\varphi_e)_{e\in N}$.

We now give a \emph{reason for introducing $\W'\subseteq\B'$} into our framework.
We exhibit $\W'$, $\W=\W'\cap\B$, and $\p$ for which $C^{+}(\W, X, \p)\neq C^{+}(\W',X,\p)$.

A set is \emph{cofinite} if it is the complement of a finite set; otherwise, it is \emph{coinfinite}.
A coinfinite r.e.\ set~$T$ is \textbf{maximal coinfinite}\footnote{
What we call \emph{maximal coinfinite sets} are known as \emph{maximal sets} in recursion theory
 \citep[page~288]{odifreddi92}.
\citet{friedberg58} constructively proves the existence of such sets.} %
if it has only trivial supersets:
if $S$ is a coinfinite r.e.\ set satisfying $S\supseteq T$, then $S\setminus T$ is finite.

Let $\W'_{0}$ be the collection of all maximal coinfinite sets.
Since maximal coinfinite sets are nonrecursive, we have $\W_{0}=\W'_{0}\cap \B=\emptyset$,
which is not a simple game according to our definition.
We can nevertheless conclude $\W_{0} \neq \W'_{0}$ and define the core,
 obtaining $C(\W_{0},X,\p)=C^{+}(\W_{0},X,\p)=X$ for any profile~$\p$.
 Let $X=\{a,b,c\}$ and define a profile~$\p$ by $\pref=\{(a,b)\}$ for all~$i\in N$.
Then, $C(\W'_{0},X,\p)=C^{+}(\W'_{0},X,\p)=\{a,c\}\neq X$.

Let $\W'_{1}=\{S'\in\B': \textrm{$S'\supseteq T$ for some maximal coinfinite $T$}\}$.
Then we can easily show that $\W_{1}=\W'_{1}\cap \B$ consists of all cofinite sets;
therefore, $\W_{1} \neq \W'_{1}$.
Let $X=\{x_{-1}, x_0, x_1, x_2, \ldots \}$ be a countable set.
Let $T\in \B'$ be a maximal coinfinite set.
Define a profile $\p$ as follows: $\pref=\{(x_i,x_{-1})\}$ if $i\in T$; $\pref=\emptyset$ otherwise.
Then, $\xprefy$ is $\{i\}$ if $x=x_i$ for $i\in T$ and $y=x_{-1}$; it is $\emptyset$ otherwise.
So $\p$ is measurable and each player has a maximal alternative.
We have $C(\W_{1},X,\p)=C(\W'_{1},X,\p)=X$.
Also, $\{i: x\notin\max\pref\}$ is $T$ if $x=x_{-1}$; it is $\emptyset$ otherwise.
So, $C^+(\W_{1},X,\p)=X$ but $C^{+}(\W'_{1},X,\p)=X\setminus\{x_{-1}\}$.\end{example}

We finally give an example of a decision-making setting to which the framework 
and Example~\ref{recursive-re} can be applied.
Our example is one of medical treatment, but it also represents other examples of
\emph{multi-criterion decision making} such as court rulings and facility location problems.

Consider a decision support system that helps medical professionals by selecting a reasonable set of alternatives from which they can choose.
Such a system can be regarded as a \emph{solution} (like the core) which assigns 
a set of alternatives to each profile of preferences.

Here, an \emph{alternative} describes, for instance, what and how much medicine 
to prescribe to the patient and when and what operation to perform on her.
We assume that there are countably many alternatives.
This is a natural assumption if the system is implemented on a digital computer.

A \emph{player} in the framework is understood to be a \emph{criterion} 
such as sex, age, temperature, blood pressure, 
or the result of a medical examination.  A criterion may be a combination of some of these.
We assume that there are countably many criteria.
The infinite criterion model is appropriate if one cannot tell in advance how many criteria
one needs to evaluate to reach a decision.
To formalize the idea that \emph{coalitions} are identifiable, 
we assume that each coalition is \emph{recursive}; that is,
there is an algorithm to decide whether each criterion belongs to the coalition.

When $i$~is a criterion, its \emph{preference} is a provisional judgment interpreted as follows:
$x \pref y$ if and only if ``alternative $x$ is better than $y$ in terms of criterion~$i$
for the patient in question.''  So, whether $x\pref y$ is true is determined only after a patient is given.
We assume that a profile is recursive: 
there is an algorithm for deciding (for all $x$, $y$, and $i$) whether $x\pref y$.
(It follows that profiles are measurable.)
This formalizes the idea that the system should give an answer at least
 for those profiles that it can grasp.

We want the decision support system to eliminate the alternatives that are, 
for ``almost all criteria,'' worse than some other alternatives in an agenda~$B$.
In other words, we want to eliminate alternatives $x$ such that 
the set $\{ i : x \notin \max_B \pref\}$ contains ``almost all criteria.''

As a suitable solution for that purpose, we can adopt the core without majority dissatisfaction.
We also formalize the intuition that a set $S'\subseteq N$ \textbf{contains
 ``almost all criteria''} as follows: $S'\supseteq T$ for some \emph{maximal coinfinite} set~$T$
(defined in Example~\ref{recursive-re}).
According to this formalization, if $S'$ contains ``almost all criteria,'' then
(it does not necessarily mean that it contains all but finitely many criteria, but)
one cannot algorithmically generate infinitely many criteria 
not belonging to this set (Appendix~\ref{maximal_compl_no_re}).

Suppose, for the time being, that one can assign winning/losing status 
only to (recursive) coalitions.
One can naturally define that a coalition $S$ is \textbf{winning} if it contains
``almost all criteria''; that is, $S\supseteq T$ for some maximal coinfinite~$T$.
\emph{A problem is that the set $\{ i : x \notin \max_B \pref\}$ may 
\textup{(i)}~contain ``almost all criteria''
without 
\textup{(ii)}~containing a winning coalition.  That is, we want to eliminate~$x$,
but the definition (Definition~\ref{c+}) of the core without majority dissatisfaction
does not allow us to eliminate it.}
Indeed, Example~\ref{recursive-re} shows that when the set itself is maximal coinfinite, 
it contains no winning coalition.

To rectify this problem, we assign winning/losing status to every r.e.\ set of criteria.
(An \emph{r.e.\ set} is half-identifiable in the sense that there is an algorithm
 that enumerates its elements.)
We define that an r.e.\ set is \textbf{winning} if it contains ``almost all criteria''
as above.  
We then have the equivalence of the following for r.e.\ sets $S'$:
(i)~$S'$ contains ``almost all criteria'';
(ii)~$S'$ contains a winning r.e.\ set;
(iii)~$S'$ is winning.
Fortunately, when $B$ is r.e., the set $\{ i : x \notin \max_B \pref\}$ is also r.e., 
since it can be rewritten as $\{ i : \exists y [y \in B \wedge y\pref x] \}$.
So to determine whether to reject~$x$ here, we only need to check whether this set is winning.

\subsubsection{The results for the extended framework}

Before stating the main result for the extended framework, we need to extend the notion of
the Nakamura number.

Let $\B'\supseteq \B$ be a collection that includes $ \B$, the Boolean algebra of coalitions.
Let $\W'\subseteq \B'$ be a collection of winning sets.
A nonempty collection $\Z\subseteq\B$ is a \emph{($\B$)-cover} of $S'\in \B'$ if 
$\bigcup \Z:=\bigcup_{S\in \Z}S \supseteq S'$.
If $\Z$ is a finite cover, then $\bigcup \Z\in\B$, since $\B$ is a Boolean algebra.
So, we can assume without loss of generality that $\#\Z$ is either $1$ or infinite.
Let $M(\W')$ be the collection of pairs $(\Y,\Z)$ such that\\
  (a)~$\Y\subseteq \W'$ and \\
  (b)~$\Z\colon \Y\toto \B$ is a correspondence that maps each winning set $W\in \Y$ to
   a $\B$-cover $\Z(W)$ of $W$
  and that satisfies  
  $\bigcap_{W\in \Y}\bigcup \Z(W)=\emptyset$.
  
\begin{definition} \label{kappa}
The \textbf{kappa number} $\kappa(\W')$ of a collection $\W'$ of winning sets relative to~$\B$
is $+\infty$ (greater than any cardinal number) if $M(\W')=\emptyset$;
otherwise, it is the cardinal number given by\footnote{
The right-hand side is well-defined, since every \emph{set} 
of ordinals has a supremum \citep[page~141]{hrbacek-j84} as well as the least element.
(The same is not true for \emph{classes} that are not sets, like the class of \emph{all} cardinal numbers.)
To see the supremum is a cardinal number, let $\sup_\gamma \alpha_\gamma=\alpha$,
where each $\alpha_\gamma$ is a cardinal number.
It suffices to show that $\alpha$ is an initial ordinal.
Suppose $\beta<\alpha$.
Then, by the definition of the supremum, $\beta<\alpha_\gamma$ for some~$\gamma$.
Since $\alpha_\gamma$ is a cardinal, $\#\beta<\alpha_\gamma$.
But $\alpha_\gamma\le \alpha$ implies that $\alpha$ is not equipotent to~$\beta$.} %
 \[ \kappa(\W')=\min_{(\Y,\Z)\in M(\W')} \max\{\#\Y, \sup\{\#\Z(W): W\in \Y\}\}.\]
\end{definition}

There is another, obvious extension $\nu'$ (defined for collections of winning sets) 
of the Nakamura number~$\nu$ (defined for simple games).
Let $\nu'(\W')$ be defined exactly as before, with $\W$ replaced by $\W'$.

\begin{lemma} \label{nu-kappa-nu}
$2\le \nu'(\W')\leq \kappa(\W')\leq \nu(\W)$ for $\W=\W'\cap \B$.
\end{lemma}

\begin{proof}
$2\le \nu'(\W')$ is a consequence of the assumptions 
$\emptyset\notin \W'$ and $\W'\neq \emptyset$.

We show $\nu'(\W')\le \kappa(\W')$ next.
If $\kappa:=\kappa(\W')=+\infty$, the result is obvious.
So, suppose otherwise.
By Definition~\ref{kappa},
there exists $(\Y,\Z)\in M(\W')$ satisfying (a) and~(b) in the definition of $M(\W')$
such that
\[
\kappa = \max\{\#\Y, \sup\{\#\Z(W): W\in \Y\}\}\ge \#\Y.
\]
Since $W\subseteq \bigcup\Z(W)$ for all $W\in \Y\subseteq\W'$,
we have $\bigcap_{W\in\Y}W \subseteq \bigcap_{W\in Y}  \bigcup\Z(W)=\emptyset$.
It follows from the definition of $\nu'(\W')$ that $\nu'(\W')\le \#\Y\le \kappa$.

We show $\kappa(\W')\leq \nu(\W)$ finally.
Suppose $\nu:=\nu(\W) \neq +\infty$ as above.
Then, by the definition of $\nu$,
there is a collection $\Y\subseteq\W$ such that $\#\Y=\nu\ge 2$ and $\bigcap_{W\in\Y}W=\emptyset$.
For each $W\in \Y$, let $\Z(W)=\{W\}\subseteq\B$.
We claim that $(\Y,\Z)\in M(\W')$.
((a)~$\Y\subseteq \W\subseteq\W'$.
(b)~$\Z(W)$ is a cover of $W$ since $\bigcup \Z(W)=W$.
Also, $\bigcap_{W\in\Y}\bigcup \Z(W)=\bigcap_{W\in\Y} W=\emptyset$.)
It follows from Definition~\ref{kappa} that
$\kappa(\W')\leq \max\{\#\Y, \sup\{\#\Z(W): W\in \Y\}\}=\#\Y=\nu$.\end{proof}

The inequalities in Lemma~\ref{nu-kappa-nu} are sometimes strict.
In the following example, we have
$\nu'(\W')<\kappa(\W')<\nu(\W)$.

\begin{example} \label{nu-kappa-nu2} 
Consider $N=[0,1]$, the unit interval on the set of real numbers.
Let $\B'$ be the collection of all subsets of~$N$.
Let $\W'$ be the collection of all dense subsets~$D$ of~$N$, that is,
between any two distinct points in~$N$, there is an element of~$D$.
For example, both $W=\{r \in N:  \textrm{$r$ is rational}\}$ and $W'=\{\sqrt{2}+r \in N: \textrm{$r$ is rational} \}$ 
belong to~$\W'$.
We have $\nu'(\W')=2$, since $W\cap W'=\emptyset$.
Let $\B$ be the collection of all finite or cofinite sets $S\subseteq N$.
Let $\W$ be the collection of all cofinite sets $S\subseteq N$.

We prove that $\W=\W'\cap \B \subsetneq \W'$.
($\subseteq$):  Suppose $S\in\W$.  Then, $S\in\B$.
Since $S^c$ is finite, $S$ is a dense subset of $N$.
($\supseteq$):  Suppose $S\in\W'\cap \B$.
Since $S\in \B$, either $S$ is finite or cofinite.
If $S$ is finite, it is not dense in~$N$.  So $S$~must be cofinite.

We prove that $\nu(\W)=2^{\omega}$,  the cardinality of the continuum.
By Lemma~\ref{nakamura-bounds}, it suffices to show that $\bigcap_{\alpha<\nu}S_\alpha=\emptyset$ for
a collection $\{S_\alpha\in \W: \alpha<\nu\}$ implies $\nu\ge 2^{\omega}$.
Taking the complement,
$2^{\omega}=\#N=\#(\bigcup_{\alpha<\nu}S_\alpha^c)
\le \nu\cdot \sup\{\# S_\alpha^c: \alpha< \nu\}$ \citep[Theorem~1.3, page~188]{hrbacek-j84}.
Since the supremum is at most~$\omega$ (because $S_\alpha^c$ is finite),
we have $\nu\ge 2^{\omega}$.

Finally, we prove that $\kappa(\W')=\omega$.

We first show that $\kappa(\W')\le \omega$. 
Let $\Y=\{W, W'\}$, where $W$, $W'\in\W'$ are the winning sets given above.
Let $\Z(W)=\{\{r\}: r\in W\}\subset\B$ and 
$\Z(W')=\{\{r'\}: r'\in W'\}\subset\B$.
It is straightforward to show that $(\Y, \Z)\in M(\W')$.
Since $\#\Y=2$ and $\sup\{\#\Z(W), \#\Z(W')\}=\omega$,
we have $\kappa(\W')\le \omega$.

We next show that $\kappa(\W')$ is not finite.
Suppose it is finite.  Then, there exists $(\Y,\Z)\in M(\W')$
such that $\Y$ is finite and $\Z(W)\subset\B$ is finite for each $W\in\Y$.
Since $\B$~is a Boolean algebra, this implies that $\bigcup \Z(W)\in \B$,
which in turn implies that $\bigcup \Z(W)$ is either finite or cofinite.
Suppose $\bigcup \Z(W)$ is finite for some $W\in\Y\subseteq\W'$.
Then  $\bigcup \Z(W)\supseteq W$ implies that $W\in \W'$ is finite,
a contradiction.
Hence $\bigcup \Z(W)$ is cofinite for all $W\in\Y$.
Being a finite intersection of cofinite sets, $\bigcap_{W\in\Y}\bigcup Z(W)$ is nonempty,
contrary to assumption.\end{example}

\begin{lemma} \label{kappa-extends-nu}
$\kappa$ is an extension of~$\nu$.  That is, if $\W'\subseteq \B\subseteq \B'$,
then $\kappa(\W')=\nu(\W')$.\end{lemma}

\begin{proof}
Suppose $\W'\subseteq \B$.
Then, $\W=\W'\cap\B=\W'$.
Lemma~\ref{nu-kappa-nu} implies $\nu'(\W')\leq \kappa(\W')\leq \nu(\W')$.
Since $\nu'$ extends $\nu$ and $\W'$ is a simple game in this case, we have $\nu'(\W')=\nu(\W')$,
from which the conclusion follows.\end{proof}

We now give the main result for the extended framework.

\begin{theorem} \label{nakamura-max-extended}
Let $\W'\subseteq \B'$ be a collection of winning sets and $B\subseteq X$ an agenda.
Let $\profs$ be the set of measurable profiles of preferences having a maximal element of~$B$.
Then the following two statements are equivalent:\\
\textup{(i)}~$\#B<\kappa(\W')$;\\
\textup{(ii)}~the core $C^+(\W', B, \p)$ without majority dissatisfaction is nonempty for  all $\p\in \profs$.\\
Moreover, these equivalent statements imply, but are not implied by\\
\textup{(iii)}~the core $C(\W', B, \p)$ is nonempty for all $\p\in \profs$.
\end{theorem}

\begin{proof}\footnote{Lemmas~\ref{bet-c+c} and \ref{nu-kappa-nu} are not useful for proving 
Theorem~\ref{nakamura-max-extended} from Theorem \ref{nakamura-max}, and vice versa.}
(i)$\Rightarrow$(ii).
Suppose $C^+(\W',B,\p)=\emptyset$ for some profile $\p\in \profs$.
We show that $\#B\geq \kappa(\W')$.

By the definition (Definition~\ref{c+}) of $C^+$,
for each $x\in B$, there is a winning set $W_x\in\W'$ such that $W_x\subseteq \xnotmax$.
We claim that $\bigcap_{x\in B}\xnotmax=\emptyset$.
(Otherwise, there is a player $i$ whose preference has no maximal element of~$B$.)

Let $\Y=\{W_x: x\in B\}\subseteq \W'$.  We have $\#\Y\leq \#B$.
For each $W_x\in\Y$, let
\[ \Z(W_x)=\{ \yprefx: y\in B\}. \]
  We have $\#\Z(W_x)\leq \#B$.

We claim that $(\Y,\Z)\in M(\W')$.
(\emph{Details}. We verify (b) of the definition of $M(\W')$.
First, $\Z(W_x)\subseteq \B$ since $\p$ measurable implies $\yprefx\in \B$.
Second, $\Z(W_x)$ is a cover of $W_x$ since 
$\bigcup Z(W_x)=\bigcup_{y\in B}\yprefx = \xnotmax \supseteq W_x$.
Third, $\bigcap_{W_x\in \Y}\bigcup\Z(W_x) = \bigcap_{x\in B} \xnotmax = \emptyset$ 
by the claim above.)

By the definition (Definition~\ref{kappa}) of $\kappa$, we get
\[ 
\kappa(\W')\leq \max\{\#\Y, \sup\{\#\Z(W_x): x\in B\} \} \leq \#B.
\]

(ii)$\Rightarrow$(i).
Suppose $\#B\ge \kappa:=\kappa(\W')$.
We construct a profile $\p\in\profs$ such that $C^+(\W',B,\p)=\emptyset$.
By the definition (Definition~\ref{kappa}) of $\kappa(\cdot)$,
there exists $(\Y,\Z)\in M(\W')$
such that\footnote{%
We can write $\Y$ and $\Z(L_\alpha)$ in the form below (footnote~\ref{note:indexed-collection}).}
\\
(a)~$\Y = \{L_\alpha: \alpha<\kappa'\}\subseteq \W'$, where $\kappa':=\#\Y$:\\
(b)~$\Z$ maps each $L_\alpha\in\Y$ to $\Z(L_\alpha) = \{L_\alpha^\beta: \beta<\beta(\alpha)\}\subseteq \B$
(where $\beta(\alpha):=\#\Z(L_\alpha)$) satisfying
$L'_\alpha:=\bigcup\Z(L_\alpha) = \bigcup_{\beta<\beta(\alpha)}L_\alpha^\beta \supseteq L_\alpha\in \W'$
and $\bigcap_{L_\alpha\in\Y} \bigcup\Z(L_\alpha) = \bigcap_{\alpha<\kappa'} L'_\alpha = \emptyset$;\\
(c)~$\kappa=\max\{\kappa', \sup\{\beta(\alpha): \alpha<\kappa' \}\} \le \#B$.

Write $B=\{x_\alpha: \alpha<\#B\}$ and let $B'=\{x_\alpha: \alpha<\kappa'\}$.

\medskip

\emph{Case}: $\kappa$ is finite. 
We construct a profile $\p$ such that the dominance relation~$\dom$ has a cycle consisting of
$\kappa$ alternatives.
Since $\beta(\alpha)\le \kappa$ is finite for all $\alpha<\kappa'$ in this case, 
$L'_\alpha=\bigcup_{\beta} L_\alpha^\beta$ is a finite union of elements of the Boolean algebra~$\B$.
So we can assume $L'_\alpha \in \B$ and $\beta(\alpha)=1$ without loss of generality.
By~(c), we have $\kappa'=\kappa$.
Write $x_\kappa=x_0$ and fix a relation $\succ'=\{(x_0,y): y\notin B' \}$.
Let 
\[
\pref=\{(x_{\alpha+1},x_\alpha): L'_\alpha\ni i\} \,\cup \succ'
\]
 for all $i\in N$.
The profile~$\p$ is the same as that in the proof of Theorem~\ref{nakamura-max},
except that $L_\alpha$ is replaced by $L'_\alpha\in\B$.
The rest of the proof runs as before.

\medskip

\emph{Case}: $\kappa$ is infinite.
We construct a profile $\p$ such that for each alternative $x_\alpha\in B$, there is a winning set of
players~$i$ who prefer another alternative $x_{\alpha+\beta_i}\in B$ to $x_\alpha$.\footnote{%
Though not required in our setting, we construct the profile so that individual preferences will be
acyclic.  This is achieved by the following:
no player prefers an alternative $x_{\alpha'}\in B$ with smaller index to $x_\alpha$
(if $\alpha'\le \alpha$, then $x_{\alpha'}\not\pref x_\alpha$).}
Since preferences involving alternatives outside~$B$ are irrelevant,
 we construct the profile in such a way that it satisfies $\pref\subseteq B\times B$ for all~$i$.

We define $\p$ by specifying $\xprefy$ for each pair $(x,y)=(x_{\alpha'},x_\alpha)\in B\times B$
 satisfying $\alpha'>\alpha$.  
(If $(x,y)$ does not satisfy the conditions concerning the indices, then $\xprefy=\emptyset$.)
Note that each such pair can be uniquely written as $(x,y)=(x_{\alpha+\beta},x_\alpha)$
for some $\beta>0$.\footnote{%
$\beta$~is the unique ordinal isomorphic to $\{\gamma: \alpha\le \gamma< \alpha'\}$.}
Let $L_\alpha^1=\emptyset$ if $\beta(\alpha)=1$; otherwise, we can assume $\beta(\alpha)$ is infinite.
Define $\p$ by the following, which immediately establishes that $\p$~is measurable:\footnote{
See Appendix~\ref{appendix-figs} for a graph representation of the profile.
For the subsequent argument, we need to verify that $x_{\alpha+\beta}$ and~$x_\alpha$ 
corresponding to the first two cases do exist in~$B$, that is, 
$\alpha+\beta<\#B$ (and $\alpha<\#B$, which is obvious).
For the first and the third cases, since $\alpha<\#B$, we have $\alpha+\beta=\alpha+1<\#B$.
This is because $\#B$, being an infinite cardinal number,
 is a limit ordinal by Lemma~\ref{lem:card-lim}.
For the second case, since $\alpha<\kappa'\le \#B$, we have $\#\alpha< \#B$.
This is because $\#B$, being a cardinal number,
is not equipotent to $\alpha<\#B$.
Similarly, since $\beta<\beta(\alpha)\le \#B$, we have $\#\beta< \#B$.
If $\alpha$ and $\beta$ are finite, the conclusion is obvious since $\#B$ is infinite in this case.
Suppose otherwise.
Then, by Lemma~\ref{lem:sumsum}, we have 
$\#(\alpha+\beta)=\#\alpha+\#\beta= \max\{\#\alpha,\#\beta\}<\#B=\#(\#B)$.
This implies that there is no one-to-one function from $\#B$ to $\alpha+\beta$.
Therefore, $\alpha+\beta<\#B$.} 
\[
\{i:x_{\alpha+\beta}\pref x_\alpha\} =
	\left\{ \begin{array}{ll}
	L_\alpha^0\cup L_\alpha^1 \in\B & \mbox{if $\alpha<\kappa'$ and $\beta=1$},\\
	L_\alpha^\beta \in \B  & \mbox{if $\alpha<\kappa'$ and $1<\beta<\beta(\alpha)$}, \\
	N   \in \B   & \mbox{if $\kappa' \le \alpha<\#B$ and $\beta=1$}, \\
	\emptyset  \in \B & \mbox{otherwise}.
	\end{array}
	\right.
\]

For each $x_\alpha\in B$, we have 
$\{i: x_\alpha\notin \max_B\pref\}=\bigcup_\beta \{i: x_{\alpha+\beta}\pref x_\alpha\}$,
which equals $\bigcup_{\beta<\beta(\alpha)} L_\alpha^\beta=L'_\alpha$ (if $\alpha<\kappa'$;
the equality holds whether $\beta(\alpha)$ is $1$ or infinite) 
or $N$ (otherwise).
In either case, the set contains a winning set $L_\alpha\in\W'$.
Therefore, $x_\alpha\notin C^+(\W',B,\p)$.  This establishes that $C^+(\W',B,\p)=\emptyset$. 

To establish that $\p\in\profs$, we need to show that each~$i$'s preference~$\pref$
has a maximal element of~$B$.
Since $\bigcap_{\alpha<\kappa'} L'_\alpha = \emptyset$, we have
$\bigcup_{\alpha<\kappa'} (L'_\alpha)^c = N$.
So, for each $i\in N$, there is an $\alpha<\kappa'$
such that $i\notin L'_\alpha$.
Since $i\notin L_\alpha^\beta$ for any $\beta<\beta(\alpha)$,
we have $x_\alpha\in \max_B\pref$ by the definition of~$\p$.

\medskip

(ii)$\Rightarrow$(iii).  Immediate from Lemma~\ref{bet-c+c}.

\medskip

(iii)$\not\Rightarrow$(i).  Consider Example~\ref{nu-kappa-nu2}.
We first prove that $C(\W',B,\p)=C(\W,B,\p)$ for all $\p\in\profs$ and all $B\subseteq X$.
By Lemma~\ref{bet-c+c}, it suffices to show that $C(\W,B,\p)\subseteq C(\W',B,\p)$.
Suppose $x \in B$ is not in $C(\W',B,\p)$.
Then, there are $y \in B$ and $S \in \W'$ such that
$S\subseteq \yprefx$.
Since $\p$ is measurable, $\yprefx\in\B$ is either finite or cofinite.
If it is finite, then $S\in \W'$~is finite, 
a contradiction.
It follows that $\yprefx$ is cofinite, hence it belongs to~$\W$.
So $x\notin C(\W,B,\p)$.

Choose an agenda $B$ satisfying $\#B=\omega$.
(i) is violated since $\#B=\omega\not< \omega=\kappa(\W')$.
On the other hand, since $\#B=\omega< 2^\omega=\nu(\W)$,
Theorem~\ref{nakamura-max} implies that $C(\W,B,\p)\ne \emptyset$ for all $\p\in \profs$.
Then (iii) is satisfied since $C(\W',B,\p)=C(\W,B,\p)$.\end{proof}


\pagebreak

\setcounter{page}{1}

\appendix

\section{Appendix: Supplementary Material}

This is supplementary material for 
``Preference aggregation theory without acyclicity: The core without majority dissatisfaction''
by Masahiro Kumabe and H.  Reiju Mihara.

\subsection{Proof of Lemma~\ref{lem:card-lim}} \label{card-lim}

We show that an infinite cardinal number is a limit ordinal.

Suppose $\alpha$ is an infinite cardinal number that is not a limit ordinal.
Then, being a successor ordinal, $\alpha=\beta+1=\beta\cup\{\beta\}$ for some ordinal $\beta$.
Clearly, $\beta<\alpha$.
Since $\alpha$ is initial, there is no bijection between $\alpha=\beta\cup\{\beta\}$ and $\beta$.
But we can construct such a bijection $f$ as follows:
$f(\gamma)=0$ if $\gamma=\beta$; $f(\gamma)=\gamma+1$ if $\gamma\in\omega$; 
$f(\gamma)=\gamma$ if $\gamma\in\beta$ but $\gamma\notin\omega$.

\subsection{Proof of Lemma~\ref{lem:sumsum}} \label{sumsum}

We show that $\#(\alpha+\beta)=\#\alpha+\#\beta$, where the sum on the left side is the ordinal sum and
the sum on the right is the cardinal sum.

Pick disjoint well-ordered sets $(A,<_A)$ and $(B,<_B)$
isomorphic to ordinals $\alpha$ and $\beta$, respectively.
Then, we have $\#(\alpha+\beta)=\#(A\cup B)$ \citep[Theorem~5.5, page~152]{hrbacek-j84}.
This is equal to $\#\alpha+\#\beta$ by the definition of the cardinal sum.

\subsection{The inclusion of Lemma~\ref{c+in_maximals} is strict for some profile} \label{ex:core3}

We show that there is a profile $\p$ such that there is an alternative not in
$C^+(\W,B,\p)$ but in $C(\W,B,\p)\cap(\bigcup_i \max_B \pref)$.

We ``replicate'' sufficiently many times an example for which $C^+$~is a proper subset of~$C$, and then
 add an ``insignificant'' player whose maximal alternative contains an alternative belonging to the difference $C\setminus C^+$.
We build on Example~\ref{ex:core1} here.
 Let $N=\{1,1',2,2',3,3',4\}$ and let $\W$ consist of the coalitions having a majority.
Let $X=\{a,b,c,d,e\}$.  Define a profile by 
$\succ_1=\succ_{1'}=\{(a,d),(e,b),(e,c)\}$,
$\succ_2=\succ_{2'}=\{(b,d),(e,a),(e,c)\}$,
$\succ_3=\succ_{3'}=\{(c,d),(e,a),(e,b)\}$, and
$\succ_4=\emptyset$ (or any preference such that $d$ is a maximal).
Then, as before, the core $C$ is $\{d, e\}$ and the core $C^+$ without majority dissatisfaction is $\{e\}$.
So $C^+$ is a proper subset of $C\cap(\bigcup_i \max\succ_i)=\{d,e\}$.

\subsection{The complement of a maximal coinfinite set contains no infinite r.e.\ sets}
 \label{maximal_compl_no_re}

Let $T$ be a maximal coinfinite set.  We show that $T^c$  contains no infinite r.e.\ subsets.

Let $S'\subseteq T^c$ be an infinite r.e.\ set.
Then $S = S' \cup T$ is an r.e.\ set containing $T$.
By the definition of maximal coinfinite sets, either $S\setminus T$ is finite or $S$ is cofinite.
But the first case cannot occur since $S \setminus T = S'$ is infinite.
Since $S$ is cofinite, we have $S=S' \cup T = F^c$ for some finite~$F$.
It follows that $T^c = S' \cup F$.
So $T^c$ is r.e., implying that $T$ is recursive.
This contradicts the fact that $T$~is maximal coinfinite.

\subsection{Figures}
\label{appendix-figs}

\subsubsection{Example~\ref{ex:core1}}

\begin{center}
\includegraphics[width=5cm,clip]{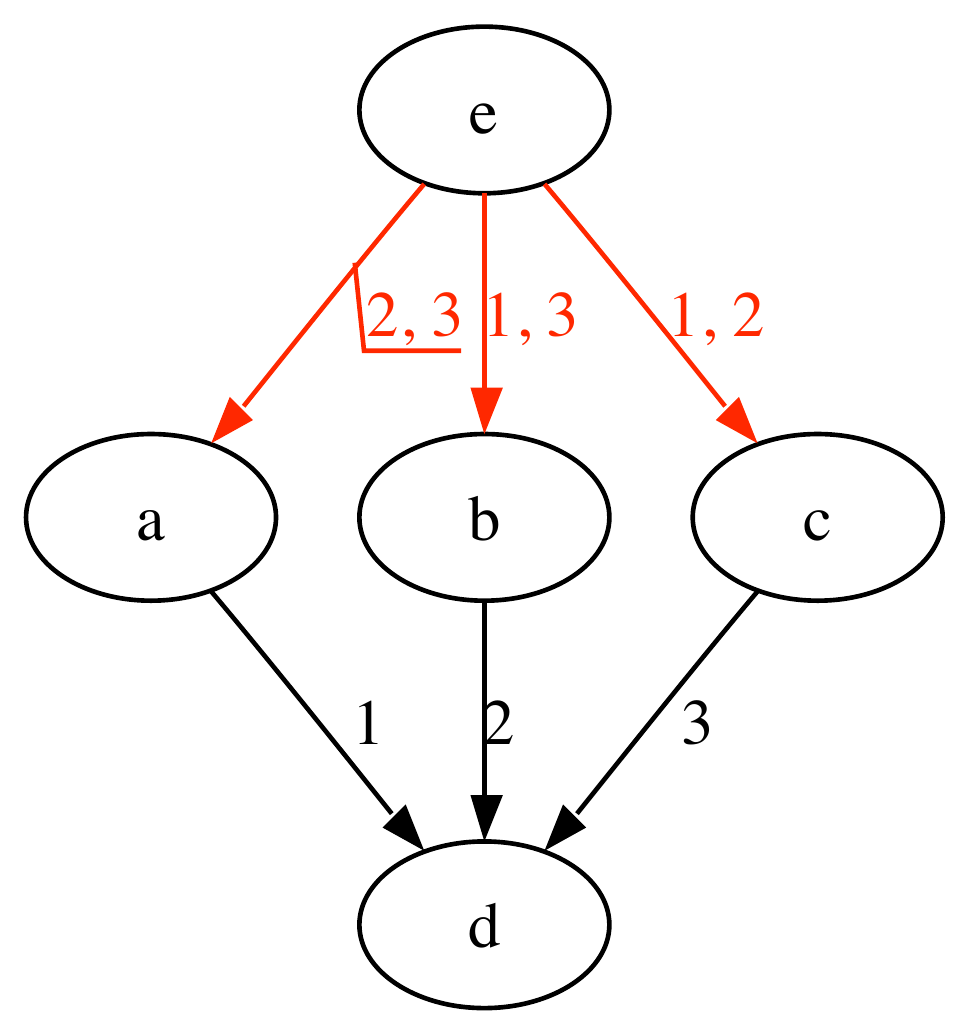}
\end{center}

\subsubsection{Example~\ref{ex:core2}}

\begin{center}
\includegraphics[width=5cm,clip]{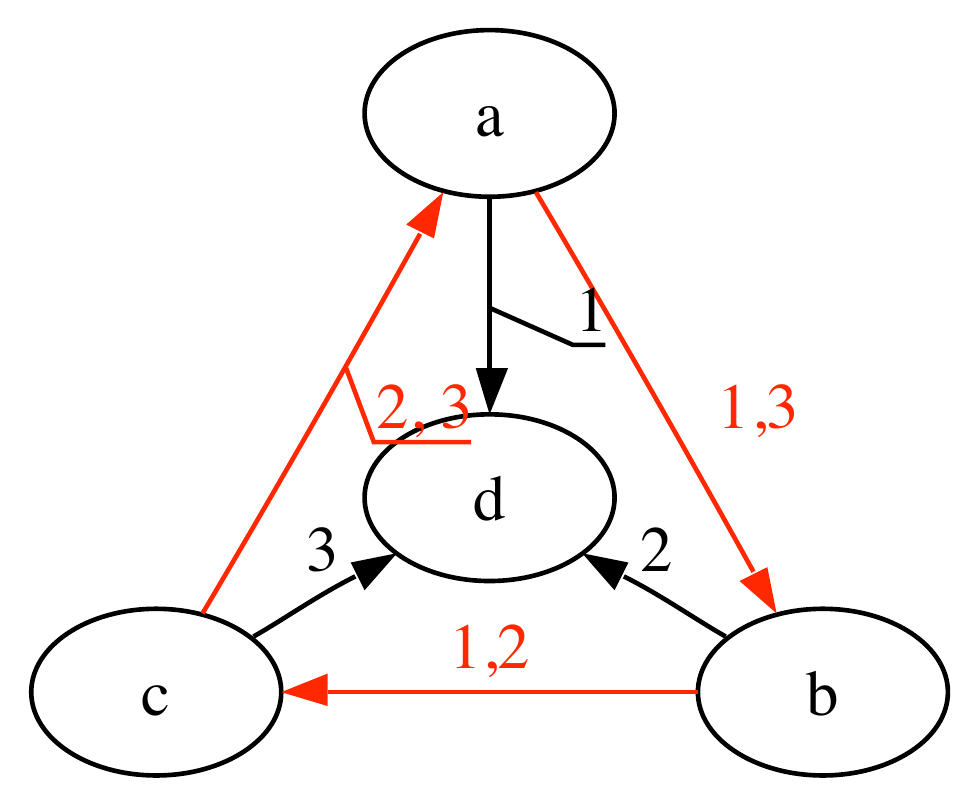}
\end{center}

\pagebreak

\subsubsection{Profile for the proof (ii)$\Rightarrow$(i) of Theorem~\ref{nakamura-max-extended}, 
when $\kappa$ is infinite}

\begin{center}
\includegraphics[width=15cm,clip]{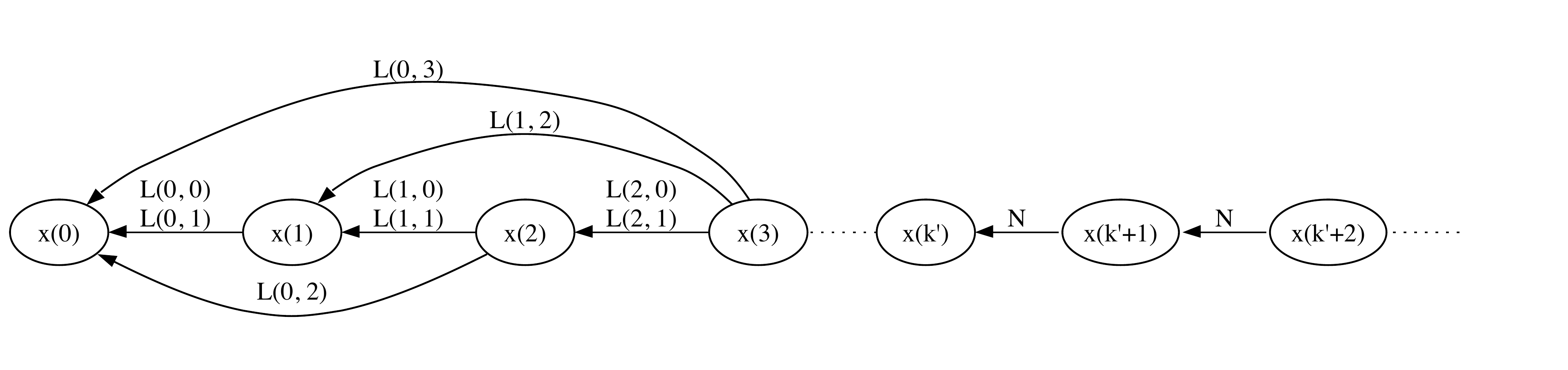}
\end{center}
For each $\alpha,\beta<\# B$, x($\alpha$) denotes $x_\alpha$,
L($\alpha,\beta$) denotes $L_\alpha^\beta$, and 
k' denotes~$\kappa'$.


\begin{thebibliography}{29}
\expandafter\ifx\csname natexlab\endcsname\relax\def\natexlab#1{#1}\fi
\expandafter\ifx\csname url\endcsname\relax
  \def\url#1{\texttt{#1}}\fi
\expandafter\ifx\csname urlprefix\endcsname\relax\def\urlprefix{URL }\fi

\bibitem[{Ambrus and Rozen(2008)}]{ambrus-r0807}
Ambrus, A., Rozen, K., Jul. 2008. Revealed conflicting preferences. Discussion
  Paper 1670, Cowles Foundation, Yale University, New Haven.

\bibitem[{Andjiga and Mbih(2000)}]{andjiga-m00}
Andjiga, N.~G., Mbih, B., 2000. A note on the core of voting games. Journal of
  Mathematical Economics 33, 367--372.
  
 \bibitem[{Andjiga and Moulen(1989)}]{andjiga-m89}
Andjiga, N.~G., Moulen, J., 1989. Necessary and sufficient conditions for
  $l$-stability of games in constitutional form. International Journal of Game
  Theory 18, 91--110.

\bibitem[{Andjiga and Moyouwou(2006)}]{andjiga-m06}
Andjiga, N.~G., Moyouwou, I., 2006. A note on the non-emptiness of the
  stability set when individual preferences are weak orders. Mathematical
  Social Sciences 52, 67--76.

\bibitem[{Arrow(1963)}]{arrow63}
Arrow, K.~J., 1963. Social Choice and Individual Values, 2nd Edition. Yale
  University Press, New Haven.

\bibitem[{Austen-Smith and Banks(1999)}]{austensmith-b99}
Austen-Smith, D., Banks, J.~S., 1999. Positive Political Theory I: Collective
  Preference. University of Michigan Press, Ann Arbor.

\bibitem[{Banks(1985)}]{banks85}
Banks, J.~S., 1985. Sophisticated voting outcomes and agenda control. Social
  Choice and Welfare 1, 295--306.

\bibitem[{Banks(1995)}]{banks95}
Banks, J.~S., 1995. Acyclic social choice from finite sets. Social Choice and
  Welfare 12, 293--310.

\bibitem[{Banks et~al.(2006)Banks, Duggan, and Le~Breton}]{banks-dl06}
Banks, J.~S., Duggan, J., Le~Breton, M., 2006. Social choice and electoral
  competition in the general spatial model. Journal of Economic Theory 126,
  194--234.

\bibitem[{Dasgupta and Maskin(2008)}]{dasgupta-m08}
Dasgupta, P., Maskin, E., 2008. On the robustness of majority rule. Journal of
  the European Economic Association 6, 949--973.

\bibitem[{Duggan(2007)}]{duggan07}
Duggan, J., 2007. A systematic approach to the construction of non-empty choice
  sets. Social Choice and Welfare 28, 491--506.

\bibitem[{Fishburn(1970)}]{fishburn70}
Fishburn, P.~C., 1970. Arrow's {I}mpossibility {T}heorem: Concise proof and
  infinite voters. Journal of Economic Theory 2, 103--6.

\bibitem[{Friedberg(1958)}]{friedberg58}
Friedberg, R.~M., 1958. Three theorems on recursive enumeration. {I}.
  decomposition, {II.} maximal set, {III.} enumeration without duplication.
  Journal of Symbolic Logic 23, 309--316.

\bibitem[{Hrbacek and Jech(1984)}]{hrbacek-j84}
Hrbacek, K., Jech, T., 1984. Introduction to Set Theory, second, revised and
  expanded edition. Marcel Dekker, New York.

\bibitem[{Kalai et~al.(2002)Kalai, Rubinstein, and Spiegler}]{kalai-rs02}
Kalai, G., Rubinstein, A., Spiegler, R., 2002. Rationalizing choice functions
  by multiple rationales. Econometrica 70, 2481--2488.

\bibitem[{Kreps(1979)}]{kreps79}
Kreps, D.~M., 1979. A representation theorem for ``preference for
  flexibility''. Econometrica 47, 565--577.

\bibitem[{Kumabe and Mihara(2008{\natexlab{a}})}]{kumabe-m08jme}
Kumabe, M., Mihara, H.~R., 2008{\natexlab{a}}. Computability of simple games: A
  characterization and application to the core. Journal of Mathematical
  Economics 44, 348--366.

\bibitem[{Kumabe and Mihara(2008{\natexlab{b}})}]{kumabe-m08scw}
Kumabe, M., Mihara, H.~R., 2008{\natexlab{b}}. The {N}akamura numbers for
  computable simple games. Social Choice and Welfare 31, 621--640.

\bibitem[{Le~Breton and Salles(1990)}]{lebreton-s90}
Le~Breton, M., Salles, M., 1990. The stability set of voting games:
  Classification and genericity results. International Journal of Game Theory
  19, 111--127.

\bibitem[{Lipman(1991)}]{lipman91}
Lipman, B.~L., 1991. How to decide how to decide how to \dots: Modeling limited
  rationality. Econometrica 59, 1105--1125.

\bibitem[{Martin and Merlin(2006)}]{martin-m06}
Martin, M., Merlin, V., 2006. On the characteristic numbers of voting games.
  International Game Theory Review 8, 643--654.

\bibitem[{Mihara(1997)}]{mihara97et}
Mihara, H.~R., 1997. {A}rrow's {T}heorem and {T}uring computability. Economic
  Theory 10, 257--76.

\bibitem[{Mihara(1999)}]{mihara99jme}
Mihara, H.~R., 1999. {A}rrow's theorem, countably many agents, and more visible
  invisible dictators. Journal of Mathematical Economics 32, 267--287.
  
\bibitem[{Mihara(2000)}]{mihara00scw}
Mihara, H.~R., 2000. Coalitionally strategyproof functions depend only on the
  most-preferred alternatives. Social Choice and Welfare 17, 393--402.

\bibitem[{Nakamura(1979)}]{nakamura79}
Nakamura, K., 1979. The vetoers in a simple game with ordinal preferences.
  International Journal of Game Theory 8, 55--61.

\bibitem[{Odifreddi(1992)}]{odifreddi92}
Odifreddi, P., 1992. Classical Recursion Theory: The Theory of Functions and
  Sets of Natural Numbers. Elsevier, Amsterdam.

\bibitem[{Penn(2006)}]{penn06scw27p531}
Penn, E.~M., 2006. The Banks set in infinite spaces. Social Choice and Welfare
  27, 531--543.

\bibitem[{Royden(1988)}]{royden88}
Royden, H.~L., 1988. Real Analysis, 3rd Edition. Macmillan, New York.

\bibitem[{Rubinstein(1980)}]{rubinstein80}
Rubinstein, A., 1980. Stability of decision systems under majority rule.
  Journal of Economic Theory 23, 150--159.

\bibitem[{Truchon(1995)}]{truchon95}
Truchon, M., 1995. Voting games and acyclic collective choice rules.
  Mathematical Social Sciences 29, 165--179.

\end{thebibliography}
\end{document}